\documentclass[preprint,12pt,authoryear]{elsarticle}

\usepackage{amssymb}

\usepackage{amsmath}
\usepackage{float}
\usepackage{booktabs}
\usepackage[bookmarks=true,bookmarksopen=true,colorlinks=true,linkcolor=blue]{hyperref}

\providecommand{\jelclassif}[1]
{
	\small	
	\text{\textit{JEL Classification:}} #1
}

\usepackage{threeparttable}
\usepackage{setspace}
\newcommand{\notesize}{\setstretch{1}\scriptsize} 

\usepackage{pifont}        
\newcommand{\cmark}{\ding{51}} 
\newcommand{\xmark}{\ding{55}} 
\newcommand{\appref}[1]{\hyperref[#1]{\ref*{#1}}}
  \usepackage{subcaption}   

\newtheorem{theorem}{Theorem}
\newtheorem{proposition}{Proposition}

\makeatletter
\def\@fnsymbol#1{} 
\makeatother

\begin{document}

\begin{frontmatter}



\title{Measuring the depth of multidimensional poverty with ordinal data}


\author[inst1]{Fernando Flores Tavares}

\affiliation[inst1]{organization={Department of Economics and Statistics, University of Siena, Italy},
email={Email: fernando.tavares@unisi.it}}

\date{}



\begin{abstract}
Standard multidimensional poverty gap measures are seldom applied because they require cardinal indicators. This paper proposes a positional poverty gap index that captures the depth of deprivations even when indicators are ordinal. Drawing on the fuzzy set literature, the method scores each deprived individual by their positional distance from the top of the empirical distribution of each indicator, normalized by the same distance for the most deprived individual. An axiomatic characterization shows that, once deprivation is built up additively from pairwise comparisons of achievements, this positional depth score is the only one consistent with ordinal data. Within the Alkire–Foster counting framework, the measure preserves the identification and aggregation structure and, under a fixed reference distribution, satisfies the main axiomatic properties. The result is a meaningful and policy-relevant depth index for multidimensional poverty.
\end{abstract}




\end{frontmatter}


\section{Introduction}
This paper introduces a positional poverty gap measure for multidimensional poverty, extending the Alkire–Foster counting framework to capture deprivation depth when indicators are ordinal. The contribution bridges the Alkire–Foster (AF) method \citep{AlkireFoster2011} and the Totally Fuzzy and Relative (TFR) approach \citep{CheliLemmi1995}. It provides a coherent framework for incorporating depth into multidimensional poverty measurement while preserving the identification and aggregation structure of the counting approach.

 In the unidimensional context, depth is typically measured through the poverty gap, the normalized shortfall from the poverty line, formalized by \citet{Sen1976} and generalized by \citet{Foster1984DecomposablePoverty} into a class of measures. \citet{AlkireFoster2011} extended the Foster–Greer–Thorbecke (FGT) class of measures to a multidimensional setting, capturing the extent, breadth, and depth of multidimensional poverty. The latter is measured through the multidimensional analogue of the FGT poverty gap, which captures the normalized shortfall of an individual's achievement from the deprivation cutoff.
This measure, however, requires cardinal indicators defined on a ratio scale, as the normalization of shortfalls presupposes a meaningful natural zero \citep{AlkireFoster2011}. Ordinal indicators do not convey meaningful information about the magnitude of differences between categories, while interval-scale indicators do not support the ratio comparisons required by the gap measure. In practice, empirical
applications of the AF framework rarely implement the poverty gap component because most poverty indicators are ordinal \citep{dattoma2024multidimensional}. As a result, the depth of multidimensional poverty remains largely
understudied.

To fill this gap, I build on the TFR fuzzy-set membership function, adapt and characterize it, and propose a new interpretation as a positional poverty gap that scores each deprived individual by their normalized positional distance from the top of the distribution of each indicator. This is equivalent to the share of the population with higher outcomes, divided by the same share for the most deprived individual. Depth thus increases as the individual moves closer to the worst observed outcome.

To establish the normative and informational basis of the depth score, I first characterize it formally. I show that it is not merely one among several admissible rank transformations, but the unique normalized deprivation score consistent with both an additive structure of pairwise achievement comparisons and the informational content of ordinal data. This characterization draws on the theory of relative deprivation: the score is the ordinal counterpart of Yitzhaki's relative-deprivation index \citep{Yitzhaki1979}, retaining the share of the population that is better off once the magnitudes of the shortfalls are no longer defined. I then analyze the behavior of the positional poverty gap within the AF class of measures by formalizing and proving its axiomatic properties. Holding the reference distribution fixed, I show that the resulting measure satisfies the main desirable axioms, including poverty focus, monotonicity, subgroup consistency, and decomposability, thereby ensuring predictable behavior across poverty comparisons \citep{AlkireEtAl2015}. Because positional depth is defined relative to a distribution, I also set out the conditions under which the measure is comparable across populations and over time, showing that fixing or pooling the reference distribution supports such comparisons while making explicit the trade-offs of each choice.

I then illustrate the methods with two applications. Using administrative data from Brazil, I show that individuals deprived in few indicators can still exhibit high positional depth scores within those indicators, a pattern that incidence and intensity alone miss. Among Bolsa Família beneficiaries, positional depth varies across households with similar income, so income alone does not identify those who are most severely deprived. Using cross-country data from the Multiple Indicator Cluster Surveys, I compare the positional poverty gap measure with the AF poverty gap where cardinal indicators are available. A formal correspondence result guarantees that the two rank the poor identically within each dimension, so in the data they rank individuals similarly overall, while their differences, which arise only across dimensions, reveal complementary aspects of poverty depth.

Beyond the measure and its characterization, the contributions of this paper are also practical and policy-relevant. A first contribution is to open depth measurement to a wider range of empirical settings. For instance, it allows researchers to ask, in ordinal and multidimensional settings, how deeply deprived the poor are, whether depth and intensity point to the same individuals and groups, and how the distribution of depth varies within the poor population. The global Multidimensional Poverty Index (MPI) \citep{undp2025globalmpi} and most national multidimensional poverty indices rely heavily on ordinal indicators. The proposed measure is directly applicable to these existing frameworks without requiring new data collection or changes to indicator selection.

Depth measurement, in turn, has a direct policy application. By identifying individuals with few but deep deprivations, the positional poverty gap measure can help direct resources toward those whom standard measures may give lower priority, not because they are not poor, but because their poverty is deep rather than broad. Moreover, by incorporating a relative component to the AF class of measures, the positional poverty gap index accounts for the social context of deprivation, complementing the absolute measures of incidence and intensity. 

The paper proceeds as follows. Sections \ref{section: background} and \ref{section: notation} review the background literature and introduce the notation. Section \ref{section: ind_dev} discusses the construction of indicators and the conditions required for ordinally consistent measurement. Section \ref{section:pov_indent} presents the poverty identification approach. Sections \ref{section: intensity} and \ref{section: distributional_poverty_gap} define the intensity and positional poverty gap components of the index. Section \ref{section: overall_poverty} combines them into the overall multidimensional poverty index. Section \ref{section: properties} characterizes the positional depth score and establishes the axiomatic properties of the measure. Section \ref{section: illustrations} presents two empirical illustrations, using administrative data from Brazil and cross-country survey data from the Multiple Indicator Cluster Surveys. Section \ref{section: conclusion} concludes.

\section{Background}
\label{section: background}

\subsection{Counting, fuzzy, and ordinal approaches}
Despite its success, the AF approach is not immune to criticism. Three limitations are particularly relevant. First, most applications cannot capture the depth of deprivation within each dimension, since the AF poverty-gap and squared-gap measures require cardinal indicators. Second, the additive aggregation structure does not account for cross-dimensional complementarity, that is, the possibility that the joint effect of several simultaneous deprivations exceeds the sum of their separate effects. Third, standard counting measures have limited sensitivity to how deprivation burdens are distributed across individuals.

Several contributions address the second limitation by assigning disproportionately greater weight to cumulative deprivation \citep{pattanaik2018,rippin2010,datt2019distribution}. Others address the third by making measures sensitive to how deprivation counts are distributed and associated across individuals and dimensions \citep{AabergeEtAl2019,Berenger2017}, or by distinguishing basic from non-basic attributes \citep{DhongdeEtAl2016}. Although these extensions broaden the scope of multidimensional poverty measurement, they do not capture the depth of deprivation within each indicator. This paper focuses on the first limitation. Nevertheless, its individual-level depth scores could be incorporated into alternative cross-dimensional and interpersonal aggregation functions designed to address the second and third limitations.

The first limitation has been addressed directly, but only partially. \citet{SethYalonetzky2021} study the unidimensional case and characterize a class of depth-sensitive poverty measures for a single ordinal variable. Their measures aggregate the population shares in deprived categories using ordering weights, with larger weights assigned to more deprived categories. This captures depth within the variable without treating the distances between ordinal categories as directly cardinal. However, because the ordering weights are defined for one ordinal indicator at a time, the framework does not provide a rule for aggregating depth across indicators with different category structures. A related contribution by \citet{pattanaikxu2019} considers the multidimensional case, allowing each attribute to have more than two ordered categories. Their framework maps the categories of each attribute into an increasing dimension-specific scoring function and combines these scores in an overall deprivation measure. However, the paper characterizes a broad family of admissible scoring rules rather than providing a fully specified operational procedure for measuring depth within each ordinal indicator. An operational depth measure spanning multiple ordinal indicators is therefore still missing.

 A parallel literature based on fuzzy-set approaches deals with ordinality using rankings. A foundational contribution is the distribution-based fuzzy approach of \citet{CheliLemmi1995}, whose TFR method derives membership values purely from rank positions and therefore does not require a cardinal interpretation of the data. \citet{BettiVerma1999} proposed an alternative formulation based on the complement of the Lorenz curve computed from weighted achievements, which instead requires cardinal measurement. \citet{BettiEtAl2006} showed how these two distribution-based membership functions can be combined within a unified framework. Building on this work, \citet{Betti2015Comparative} proposed a multidimensional implementation in which indicators are first transformed using the normalized TFR, enabling the consistent application of the Lorenz-curve complement across indicators. The depth index proposed here builds on this transformation approach, reinterpreting and adapting it to measure deprivation depth.
  
  Despite the advantages, fuzzy approaches have two limitations relevant here. First, without a deprivation threshold or censoring, all but those in the top category receive non-zero scores, violating the deprivation and poverty-focus axioms relative to any meaningful deprivation cutoff. Second, the overall fuzzy index yields a single aggregate measure without separating poverty into incidence, intensity, and depth, making it incompatible with the reporting framework that has made AF attractive for policy.

Two papers have attempted to bridge the fuzzy and counting approaches.
  \citet{Kobus2017} develops a fuzzy counterpart to the AF MPI
  linking union and dual-cutoff identification, and
  \citet{Gangopadhyay2020} propose a comprehensive axiomatic
  framework combining absolute deprivation with a distributional
  component. Both preserve elements of the AF structure, but both
  retain a cardinal deprivation gap as the within-dimension depth
  score \citep[in the tradition of][]{BourguignonChakravarty2003}, making them inapplicable when indicators are ordinal.

This paper draws on the counting and fuzzy-set traditions to address the depth problem raised by the literature on ordinal poverty measurement. It keeps the AF framework structure, incorporating the distribution-based score of the
  fuzzy tradition within it and formalizing the score and its axiomatic properties for ordinal indicators in a multidimensional setting.  The measure applies to mixed indicator sets combining ordinal, cardinal, and binary variables within the same index. In doing so, it provides a
  new interpretation of the score as a positional analogue of the AF poverty gap. The overall index factors into incidence, intensity, and depth, a structure
  the fuzzy literature does not provide when ordinal variables are included.

  \subsection{Relative poverty approaches}

The distinction between absolute and relative poverty is long-standing in the literature. Writing in the eighteenth century, \citet{smith1776wealth} observed that necessities include not only what is indispensable for survival, but also what custom renders as necessary for social participation. \citet{townsend1979poverty} similarly argued that poverty is inherently relative: individuals are poor when they lack the resources to participate in the customary activities and living conditions of their society\footnote{Consistent with this perspective, the EU adopts a relative definition of poverty rooted in the concept of socially acceptable living standards, originating in the 1975 Council Decision and later formalized in EU policy documents in the early 2000s \citep{council1975povertyprogramme,europeancommission2004joint,spc2001indicators}. The at-risk-of-poverty (AROP) indicator operationalizes this concept using a relative income threshold set at 60\% of the national median, while the broader AROPE framework reflects its multidimensional nature by combining income poverty with material deprivation and labor market exclusion.}. \citet{sen1983poor} refined this view by distinguishing between an absolute core in capabilities and a relative component in resources. Drawing on Smith’s reasoning, he illustrates this point by noting that the relevant capability, such as social participation, is constant across societies, while the resources required to achieve it vary: a linen shirt in eighteenth-century Britain as a condition for appearing in public without shame, and different goods fulfilling the same role in other times and places, such as housing, or access to utilities and other social infrastructure required for participation.

In the measurement literature, the distinction between absolute and relative poverty was later formalized, notably by \citet{foster1998absoluterelative}, who analyzed poverty measures under fixed and distribution-dependent poverty lines, and by \citet{ravallion2011weaklyrelative}, who proposed weakly relative poverty lines that rise with average income above an absolute floor, preserving a minimum while reflecting social context. Building on this approach, \citet{RavallionChen2019} show that the choice of the comparison income against which relative needs are assessed can materially affect global poverty measurement. In these formalizations, relativity is confined to the definition of the poverty threshold.

Extending this threshold-based approach to the multidimensional setting, \citet{dotterklasen2020absolute} apply \citet{ravallion2011weaklyrelative}'s framework within the AF multidimensional setting by linking deprivation thresholds to national standards, aiming to capture Sen’s distinction between absolute functionings and relative resources. In their approach, the relative element applies only to the thresholds of resource-like indicators, while the indicators themselves remain absolute.

Beyond poverty line definitions, the literature on relative deprivation concerns how an individual's standing is assessed in relation to others. \citet{Runciman1966} described relative deprivation as a perceived grievance arising from comparisons with a reference group of better-off persons. \citet{Yitzhaki1979} quantified this idea into an income-based index equal to the proportion of society that is richer than the individual multiplied by the average income gap separating the individual from that group. The index thus measures an objective condition, being outranked and by how much, rather than a perceived one. This index was subsequently characterized axiomatically by \citet{EbertMoyes2000} and, under a framework in which the reference group is the entire society, by \citet{BossertDambrosio2006}. \citet{Kakwani1984RelativeDeprivation} developed a distributional representation of the same measure through the relative-deprivation curve. These indices are designed for monetary poverty and require cardinal variables, because they capture not only the number of individuals who are better off, but also the magnitude of the differences separating them.

A directly related perspective is offered by \citet{CheliLemmi1995}, discussed above, who propose a fully relative measure extended to multidimensional poverty. Unlike approaches that rely on fixed deprivation thresholds, both the identification of the poor and the measurement of deprivation depend entirely on the distribution of achievements in the population. Poverty is thus determined by individuals’ relative positions rather than by comparison to an absolute standard.

 The proposed measure here embeds TFR logic within the AF framework but retains fixed deprivation thresholds, combining absolute identification with a positional assessment of depth. The cutoff determines whether a minimum standard is unmet, while the reference distribution determines how exceptional the resulting position is. 
 
The fuzzy tradition discussed above supplies the empirical form of the score, while the relative-deprivation literature supplies its formal structure, in which deprivation is built up from pairwise comparisons of achievements. Subsection~\ref{section: rd_foundation} combines that structure with ordinal invariance into a formal characterization of the positional depth score.

\section{Notation}
\label{section: notation}
Let $X$ denote the achievement matrix of size $n \times d$,
where rows index individuals $i = 1,2,\dots,n$ and columns index indicators
$j = 1,2,\dots,d$. The entry $x_{ij}$ represents the achievement of 
individual $i$ in indicator $j$.  

Each indicator $j$ is associated with a deprivation cutoff $z_j$. 
I define the deprivation status of $i$ in $j$ as
\begin{equation}
g_{ij}^0 = \mathbf{1}\{x_{ij} < z_j\},
\end{equation}
where $\mathbf{1}\{\cdot\}$ is the indicator function that equals one if the 
statement is true and zero otherwise\footnote{%
The $g_{ij}^0$ elements form the deprivation matrix $g^0$ and $z_j$ deprivations cutoffs form the vector $z$. Let
\[
g^{0} \;=\; [\,g^{0}_{ij}\,]_{n\times d}
= \begin{bmatrix}
g^{0}_{11} & \cdots & g^{0}_{1d}\\
\vdots     &        & \vdots\\
g^{0}_{n1} & \cdots & g^{0}_{nd}
\end{bmatrix}
\qquad
z=
\begin{bmatrix}
z_1\\
\vdots\\
z_d
\end{bmatrix}.
\]
}. Indicators are assigned relative importance through a weight vector 
$w=(w_1,\dots,w_d)$, with $\sum_{j=1}^d w_j = 1$.  

In addition, to capture depth I define the empirical cumulative distribution 
function (CDF) for each indicator $j$ as
\begin{equation}
F_j(x) = \frac{1}{n}\sum_{i=1}^n \mathbf{1}\{x_{ij} \leq x\},
\end{equation}
which gives the proportion of individuals whose achievement in $j$ does not 
exceed $x$. Its complement, $1-F_j(x_{ij})$, represents the share of the 
population that has strictly higher achievements than individual $i$ in 
indicator $j$\footnote{In the empirical application, a population-weighted version of $F_j$ is used.}. This transformation places all indicators on a common 
$[0,1]$ scale, enabling meaningful comparisons across different types of 
variables.

\section{Measurement conditions}
\label{section: ind_dev}
Many factors need consideration when building an indicator, such as indicator selection, unit of identification, applicable population, and others \citep[see][]{AlkireEtAl2015, Tavares2024Gender}. These factors depend on the study's objective and the researcher's judgment values. In this section, I focus on desirable conditions to build ordered indicators, apply admissible transformations, and ensure comparability, considering meaningfulness, uncertainties, and data availability\footnote{This section draws partially from insights presented in Chapter 2 of  \citet{AlkireEtAl2015}.}. 

\subsection{Indicators construction}
\label{section: 3.1}
One of the advantages of the proposed method is to use the information available in each variable to estimate their deprivation depth (positional poverty gap), even when they are ordinal. Before estimating the index, a previous step is to build indicators that establish a clear and complete ranking for each element based on data availability. Based on the variable type, one can implement different procedures to build indicators, as outlined below.\footnote{Nominal variables cannot be ordered and thus cannot serve as deprivation indicators, though they remain useful for subgroup analysis, conditional expectations, and other multi-variable models.}

\textit{Ordinal variables.} Suppose that, in a dataset, a raw variable $r \in \{1, \ldots, R\}$ represents an attribute that has $R$ categories. To build an indicator based on these elements, they must be classified into ordered achievements such that each ascending value represents a strictly higher achievement ($<$) than the previous value, and in a way that one can define a deprivation cutoff $z_j$ and non-trivially split the set into at least two parts. Let the ordered set of achievements for indicator $j$ be $\mathcal{X}_j = \{x_{j1}, \ldots, x_{jC}\}$, where $C \geq 2$ is the number of categories and, for each $e \in \{1, \ldots, C-1\}$, the relationship $x_{je} < x_{j,e+1}$ holds. Additionally, for an individual $i$, the observed achievement $x_{ij} \in \mathcal{X}_j$ is classified as deprived if $x_{ij} < z_j$ and as non-deprived if $x_{ij} \geq z_j$. When $r$ has only two categories or can only be partitioned into two ordered sets, then $C=2$.
 
However, it is not always possible to rank all elements strictly. Some categories of $r$ may be incomparable, so that neither can be judged more adequate than the other; such categories may share the same achievement value $x_{je}$, forming a semiorder, or be distinguishable only as adequate and inadequate, forming a weak order. An achievement $x_{je}$ may therefore comprise one or more categories of the raw variable. Housing materials illustrate the case: some flooring materials, such as sand and other natural materials, are not clearly differentiable in adequacy, and their ranking may depend on the climatic, geographical, and social context. When the classification of a category is unclear, a good practice is to decide based on the context and objective of the study and to run robustness checks over the alternative classifications.

The positional poverty gap focuses on the ordering within the deprived set, since achievements at or above $z_j$ are censored. Measuring depth therefore requires the deprived achievements, $x_{ij}<z_j$, to be non-trivially divided into at least two ordered subsets. When a single raw variable does not allow this, as with a binary electricity indicator, raw variables representing the same attribute can be combined, for example access to electricity with its weekly frequency. Combining variables that represent related but distinct attributes can also generate multiple ordered categories, but should be avoided when possible: such composites obscure which factor drives the deprivation score and weaken the policy interpretation.

\textit{Cardinal variables.} Raw variables $r$ that are cardinal (interval scale or ratio scale) are originally in a ranked order. Therefore, one need only to define the deprivation cutoff $z_j$ and, when necessary, change the direction of the variable from most deprived to non-deprived.

\subsection{Admissible transformations}
\label{section: 3.2}
After building indicators, the next step is to ensure that, when estimating the measures, only admissible transformations are applied, and admissible statistical and mathematical operations in the indicators considering their type (i.e., ordinal, and cardinal). In this way, the measures will maintain meaningfulness. 

Following \citet{AlkireEtAl2015}, I rely on \citet{Stevens1946}'s classification of measurement scales: nominal, ordinal, interval, and ratio, the first two qualitative and the last two cardinal. Table \ref{table:A1} in Appendix~\ref{appendix: A} reports the admissible transformations and operations for each scale, with examples relevant to poverty measurement.

As Table \ref{table:A1} shows, for ordinal variables no mathematical operation is admissible: subtraction, for instance, would treat the intervals between categories as meaningful when they are not. Types of sanitation facilities are ranked by adequacy, but the ranking carries no information on the magnitude of the improvement from one category to the next.

The poverty gap is an example of a measure that applies mathematical operations that are not admissible for ordinal indicators (i.e., subtraction and division). The normalized poverty gap therefore rests on a stronger condition than cardinality alone: it is meaningful for ratio-scale indicators, since normalizing the shortfall by the cutoff presumes a natural zero \citep{AlkireFoster2011}. An interval-scale indicator, whose zero is set by convention, admits the shortfall $z_j - x_{ij}$ but not a meaningful normalization of it. Nutrition $z$-scores, classified as interval-scale in Table~\ref{table:A1}, are a common example: obtaining a normalized gap for them requires assuming a reference floor to play the role of the natural zero. In contrast, the intensity as estimated in the AF method uses an admissible transformation: it dichotomizes the indicators into deprived and non-deprived, which is a monotonic transformation.

  The proposed positional poverty gap measure extends the AF counting framework while relying only on admissible transformations and statistics.  To meaningfully estimate the
  depth of deprivations, the positional poverty gap measure relies on the frequency distribution, which is a statistic permissible for ordinal and cardinal variable types (see Table \ref{table:A1})\footnote{Sections \ref{section: intensity} and \ref{section: distributional_poverty_gap} further elaborate on the construction of the measures.}. More specifically, it
  applies a cumulative distribution function (CDF), which converts all variables into the same unit: the share of individuals with higher achievements than those who have achieved $x_{ij}$. In this way, the resulting quantities possess ratio-scale properties, with a meaningful zero and interpretable differences and ratios.

These transformations must also yield comparability across indicators and individuals. Across indicators, the CDF transformation expresses every indicator in the same unit and range, the share of individuals with higher achievements, while preserving each variable's original ordering. Indicator or dimension weights then reflect their relative importance. Across individuals, the measure rests on the standard assumption that identical deprivation levels represent an equivalent state of poverty, as in monetary poverty measurement \citep{AlkireEtAl2015}.

\subsection{Intertemporal and spatial comparability}
\label{section: comparability_time_space}

Applying the measure across populations or time periods raises additional comparability requirements beyond those within a single country/region and cross-section. For AF's incidence and intensity, comparability across time holds when deprivation thresholds remain constant, since both are absolute measures tied to fixed thresholds. This rests on an implicit assumption: that the threshold continues to represent a meaningful deprivation cutoff across time and populations.

Like the AF poverty gap, the positional poverty gap requires harmonized thresholds and comparable values below the deprivation cutoff. The difference is that while the AF gap requires these values to be cardinal, the positional gap requires only that they be ordinally ranked. In addition, it also requires a common reference distribution $F^*_j$ because individual positional depth scores depend directly on the distribution against which they are computed.

For intertemporal analysis, two options allow the depth component to be compared across periods. The first is to fix $F^*_j$ to a reference period and hold it constant in subsequent years. The second is to pool all years into a single distribution, so that the reference reflects the entire observed trajectory.
  Table~\ref{tab:reference-options} presents the tradeoffs and normative interpretations of each option.
  
   \begin{table}[H]
  	\centering
  	\small
  	\caption{Reference distribution options for the positional poverty gap}
  	\label{tab:reference-options}
  	\setlength\tabcolsep{3pt}
  	\begin{tabular}{p{2.9cm} p{2.1cm} p{2.1cm} p{2.75cm} p{2.75cm}}
  		\hline
  		Option & Temporal comparability & Spatial comparability &
  		Normative interpretation & Trade-off \\
  		\hline
  		In-sample ($F_j$) \newline \textit{\small ranks individuals within their own sample's distribution} &
  		No &
  		No &
  		Poverty assessed relative to the current distribution (AROP-like) &
  		No design choice required; scores not comparable
  		across samples or periods \\
  		
  		Fixed ($F^*_j$) \newline \textit{\small anchors scores to a pre-specified external distribution} &
  		Yes, when anchored to a reference period &
  		Yes, when anchored to a reference population &
  		Poverty anchored to a chosen normative standard &
  		Stable and comparable within the defined scope;
  		requires an explicit normative choice of baseline \\
  		
  		Pooled ($F^*_j$) \newline \textit{\small ranks individuals relative to the combined multi-sample distribution} &
  		Yes, when the pool spans multiple periods &
  		Yes, when the pool spans multiple populations &
  		Poverty assessed relative to the pooled distribution &
  		Internally comparable across the pool; sensitive
  		to pool composition \\
  		\hline
  	\end{tabular}
  \end{table}

Spatial comparability raises analogous considerations. For incidence and intensity, harmonized thresholds and categories across geographic units are necessary, which is a normative
  choice that assumes that the same threshold represents the same deprivation level across all populations. For the positional poverty gap, it is necessary to also have harmonized ordinal values below the threshold (i.e., deprivation values\footnote{For many common poverty dimensions, internationally established ordinal classifications already exist. For instance, \citet{who2006growth} anthropometric standards define severity bands for malnutrition; the \citet{who_unicef_jmp2017}'s JMP service ladders rank water and sanitation access across five levels; \citet{uis2012isced2011} define education levels through ISCED classifications; and \citet{bhatia_angelou2015beyondconnections} and \citet{esmap2020cooking} classify electricity access and cooking fuel using the multi-tier framework.}), together with a common reference distribution. The same two options apply: one may fix $F^*_j$ to a reference population or pool all geographic units into a shared distribution.

  When both temporal and spatial comparability are required simultaneously, the user should define both a reference period and a reference population, deciding for each whether to fix or pool. When computing a class of measures, one natural approach is a hybrid implementation: the depth component is computed with an anchored $F^*_j$, while incidence and intensity are computed separately for each period and population. This preserves the spatial and inter-temporal comparability of the positional poverty gap while allowing incidence and intensity to reflect
  contemporaneous absolute deprivations.

  A different route is to calculate $F_j$ separately for each period and population (in-sample $F_j$). This choice is normatively consistent with widely used relative measures in
  official statistics. The Eurostat at-risk-of-poverty rate (AROP), for instance, sets the poverty threshold at 60\% of each country's median income in each year,
  capturing the distributional context of each setting without imposing cross-country or cross-time comparability in absolute terms. The in-sample $F_j$ follows the
  same normative logic and can be justified on those grounds.

  Section~\ref{section: properties} discusses the consequences of each choice for the axiomatic properties of the index.

\section{Poverty identification}
\label{section:pov_indent}

The identification of poverty is not trivial, especially in a multidimensional context. The approach to identifying deprivations and poverty is not only a statistical matter, but it is a choice that has potential policy implications. In this section, I explain the main poverty identification methods and how the proposed method incorporates them.

Before identifying the poor, the method records, for each individual $i$, the deprivation status $g_{ij}^0$ of each indicator $j$. As Subsection \ref{section: 3.1} shows, for each $j$ a cutoff $z_j$ is defined, such that $g_{ij}^0=1$ if $x_{ij}<z_j$ and $g_{ij}^0=0$ otherwise. The next step is calculating the individual deprivation score $c_i$, which is the sum of weighted deprivations:
\begin{equation}
c_i = \sum_{j=1}^{d} w_j g_{ij}^0.
\end{equation}
where the weights $w_j$ are positive and normalized to sum to one, $\sum_{j=1}^d w_j=1$. 

Given a poverty cutoff $k$, the identification function $\rho_i(k)$ assigns poverty status according to the individual deprivation score:
\begin{equation}
\rho_i(k) =
\begin{cases}
1 & \text{if } c_i \geq k, \\
0 & \text{if } c_i < k.
\end{cases}
\end{equation}

When the deprivation status is censored to only consider poor individuals, $g^0_{ij}(k)=\rho_i(k)g^0_{ij}$.

If $k=\min_j w_j$, $\rho_i(k)$ corresponds to the \textit{union} approach, classifying people as poor if they are deprived in at least one indicator with positive weight. If $k=1$, $\rho_i(k)$ corresponds to the \textit{intersection} approach, classifying as poor only those deprived in all indicators\footnote{\citet{Atkinson2003} first explained the distinction between union and intersection criteria and how they relate to social-welfare approaches to multidimensional deprivation.}. Finally, for intermediate values $\min_j w_j<k<1$, $\rho_i(k)$ corresponds to the \textit{dual-cutoff} approach proposed by \citet{AlkireFoster2011}, which allows flexible poverty lines between the extremes\footnote{ Beyond these count-based rules, alternative identification criteria have been proposed; \citet{Decerf2023}, for instance, develops a preference-based approach that distinguishes individuals with an extreme deprivation in a single dimension from those with moderate deprivations across several. }.

Considering the dual-cutoff approach, \citet{AlkireFoster2011,AlkireFoster2011b} classify the union and intersection identification as extreme cases, and advocate for intermediate poverty lines according to the objectives and preferences of the user, which is a way of embodying normative judgments. They recognize that poverty line choices are potentially arbitrary and, to minimize this problem, recommend robustness analysis to test for different choices. Therefore, if users want to align the method with the dual cutoff approach, they may set an intermediate poverty line such as $\min_j w_j<k<1$.

In the fuzzy literature, poverty is treated as a vague indicator, recognizing that there is no exact poverty line that can clearly classify people as poor or non-poor. To reduce arbitrariness, the approach usually dispenses with the definition of deprivation cutoffs and poverty lines by assigning poverty degrees to every unit in the sample, which is compatible with the union approach. Therefore, if users want to align the method with the fuzzy approach, they can adopt the union identification rule ($k=\min_j w_j$) and set the deprivation cutoff of each indicator at its top category, $z_j = x_{jC}$, so that only individuals in the top category are non-deprived. In practice, this specification implies that no effective deprivation and poverty cutoffs are imposed\footnote{However, under this specification the cutoffs no longer have any real effect. All but the top category are deprived, so nearly everyone is identified as poor, intensity loses its interpretation as deprivation breadth relative to a poverty standard, and the focus axioms hold only trivially, since almost no one is non-deprived (see Section~\ref{section: properties}).}.

Therefore, one may view the proposed approach here as a more general framework that encompasses both the counting and fuzzy approaches. Even under the dual-cutoff identification, it can treat poverty as a matter of degree rather than a strictly binary condition. Once the method identifies an individual as poor, the analysis can also assess the extent of that person’s deprivation. In this sense, the poverty line would separate the non-poor from the poor, while allowing those identified as poor to experience poverty to different degrees.

\section{Intensity}
\label{section: intensity}
Following \citet{AlkireFoster2011}, the measure of intensity represents the breadth of multidimensional poverty based on the number of weighted deprivations, $c_i$.  The individual intensity is defined as
\begin{equation}
A_i = \sum_{j=1}^d w_j\,g^0_{ij}(k) = c_i(k),
\end{equation}
where $c_i(k) = \rho_i(k)\,c_i$ denotes the censored deprivation score, equal to zero for the non-poor.
Since the weights $w_j$ are normalized to sum to one, $A_i$ ranges between zero and one, with greater values indicating stronger deprivation breadth. 

Finally, the aggregated intensity is measured by the average share of weighted deprivations among the poor:
\begin{equation}
A =\frac{1}{q}\sum_{i=1}^{n}\sum_{j=1}^d w_j g^0_{ij}(k) = \frac{1}{q} \sum_{i=1}^n A_i,
\end{equation}
where $q$ is the number of poor individuals.\footnote{Throughout, $q > 0$ is assumed when computing $A$ and $S$; by convention both equal zero when $q = 0$.} In this way, intensity captures the breadth of multidimensional poverty by identifying individuals with many deprivations occurring simultaneously.

\section{Positional Poverty Gap}

\label{section: distributional_poverty_gap}

The positional poverty gap measures the depth of multidimensional poverty. The proposed approach builds on the cumulative distribution function (CDF), drawing on the intuition of fuzzy poverty measures \citep{CheliLemmi1995, BettiEtAl2006}, but reinterprets it here as a positional poverty gap measure. This section formalizes the measure and clarifies its interpretation as a depth index.

Consider an individual $i$ deprived in indicator $j$. The depth of this deprivation is measured by the individual's position in the distribution of that indicator, taking the worst observed achievement as the benchmark of maximal depth. The \emph{positional depth score} is defined as
\begin{equation}\label{eq:sij}
s_{ij} = \frac{1-F_j(x_{ij})}{1-F_j(\min_i x_{ij})},
\end{equation}
  where $F_j(x_{ij}) \in [0,1]$ is the individual's positional rank in the distribution of indicator $j$, $1$ denotes the top position in that distribution, and
  $\min_i x_{ij}$ denotes the lowest achievement observed for $j$ in the reference distribution\footnote{If all individuals share the minimum value of a deprived indicator, so that $F_j(\min_i x_{ij}) = 1$, the score is set to $s_{ij}=1$ for all deprived individuals, consistent with the binary case in which all deprived individuals are equally the most deprived.}. The score gives the share of individuals better off than $i$, the distance from the top of the distribution, relative to that same distance for the most deprived. It is thus a normalized positional depth, equal to one for the most deprived individual. 

This positional notion of depth adapts the AF poverty gap to ordinal data. In the AF framework, depth is the normalized shortfall from the threshold, $g^{AF}_{ij} = (z_j - x_{ij})/z_j$, so the deprivation cutoff serves as the reference point. The positional depth score instead takes the top of the distribution as reference, a point that remains identifiable from rankings alone when shortfalls from the cutoff carry no cardinal meaning. By construction, the most deprived individual receives $s_{ij} = 1$ and non-deprived individuals are censored to zero. For binary indicators, all deprived individuals share the lowest observed achievement and therefore receive $s_{ij} = 1$.

The normalizer $1-F_j(\min_i x_{ij})$ is not an arbitrary choice. If the score is proportional to the share of individuals better off than $i$, $1-F_j(x_{ij})$, and the most deprived individual is required to receive a score of one, then the proportionality factor must equal $1/\big(1-F_j(\min_i x_{ij})\big)$. The natural alternatives fail one of these requirements: the raw share $1-F_j(x_{ij})$ leaves the scale without a fixed upper endpoint, since the most deprived individual receives $1-F_j(\min_i x_{ij})<1$, a value that drifts across indicators and samples; a threshold-based score, $\big[h_j-F_j(x_{ij})\big]/\big[h_j-F_j(\min_i x_{ij})\big]$, where $h_j$ denotes the population share below the cutoff $z_j$, is undefined for binary indicators, where the single deprived category gives $F_j(\min_i x_{ij}) = h_j$, so that both numerator and denominator equal zero. The normalizer also has a substantive reading: it takes the most deprived individual in the reference distribution as the benchmark for depth. This benchmark is identifiable from ordinal rankings alone, and it is the point of maximal depth, so fixing the scale there is what guarantees a maximal score of one for every indicator, including binary ones. 

This reading is objective: the score records how exceptional the deprived state is in the reference population. Where deprivation is rare, a high score identifies individuals left furthest behind a standard that most of the population reaches. Subsection~\ref{section: rd_foundation} derives the remaining element of this construction: the proportionality to the share of better-off individuals is itself the unique form consistent with an additive structure of interpersonal comparisons and with ordinal invariance, so the score is characterized rather than posited.

After calculating the scores, I then censor them to ensure that they only apply to deprived indicators. I define the censored positional depth score of each deprived indicator as $g^{1}_{ij}=g^{0}_{ij}\,s_{ij}$, replacing each item in the gap matrix with the indicator's positional depth score and replacing with zero when the indicator $x_{ij}\geq z_j$. I then apply the poverty identification function such that $g^{1}_{ij}(k)=g^{1}_{ij}\,\rho_i(k)$, further replacing elements with zero when an individual is not poor.

The positional depth scores are invariant to changes in the deprivation cutoff. Because the scores $s_{ij}$ are computed from the distribution of $x_{ij}$ before applying censoring, they do not depend on the choice of $z_j$. When the cutoff changes, only identification changes: individuals who move above the threshold are censored to zero, while the scores of those who remain deprived stay the same. The cutoff thus governs identification, whereas the positional depth score reflects only the individual’s position in the distribution, independently of where the threshold is set.

At the individual level, positional poverty gap is defined as the weighted average of censored positional depth scores, normalized by the individual’s censored deprivation score. Formally,
\begin{equation}
S_i = 
\frac{\sum_{j=1}^d w_j \, g^{1}_{ij}(k)}{\sum_{j=1}^d w_j g_{ij}^0(k)}.
\end{equation}
By convention, $S_i = 0$ whenever $\sum_j w_j g^0_{ij}(k) = 0$, that is, for all non-poor individuals.

 The aggregate positional poverty gap used in the index is then:
\begin{equation}
S \;=\; 
\frac{\displaystyle \sum_{i= 1}^n\sum_{j=1}^d w_j\, g^{1}_{ij}(k)}
     {\displaystyle \sum_{i= 1}^n\sum_{j=1}^d w_j\, g^0_{ij}(k)} .
\end{equation}
In words, $S$ is the average depth of a \emph{typical deprivation episode} among the poor, or, equivalently, a deprivation-weighted average of individual positional poverty gap values.

The reference distribution used to compute $F_j$ is an analytical choice open to the researcher, and each option embeds a different normative assumption (see Subsection~\ref{section: comparability_time_space} and Table~\ref{tab:reference-options}). This flexibility mirrors the role of the poverty cutoff $k$ in the AF framework: it is a normative decision, not a technical default, and should be stated, justified, and subjected to robustness analysis alongside the choice of dimensions, weights, and thresholds. Just as counting-based poverty orderings can be sensitive to these standard choices \citep{AzpitarteEtAl2020}, the reference distribution is a further factor along which rankings should be probed.

\section{Overall multidimensional poverty index}
\label{section: overall_poverty}

This section integrates the two preceding components, intensity $A_i$ and the individual positional poverty gap $S_i$,  into a single poverty degree at the individual level. Aggregating these degrees across the population yields the overall index $P$, which captures incidence, breadth, and depth of multidimensional poverty in a single summary measure.

\subsection{Individual poverty degrees}
As mentioned previously, multidimensional poverty can be interpreted in terms of degrees. The measure then reflects intermediate situations at the individual level, rather than an “all-or-nothing" poverty condition. Assessing the extent of an individual's deprivation, prior to aggregating across the population, is itself a distinct measurement step that has received comparatively little attention \citep{Permanyer2014}.

To represent intermediate situations, I define poverty degrees as the multiplication of the individual intensity and positional poverty gap measures. In formal terms, the poverty degree of individual $i$ is given by
\begin{equation}
P_i = A_i \cdot S_i .
\end{equation}

Substituting the definitions of $A_i$ and $S_i$, one obtains
\begin{equation}
P_i = \left(\sum_{j=1}^d w_j g_{ij}^0(k)\right) \cdot \left(\frac{\sum_{j=1}^d w_j g_{ij}^1(k)}{\sum_{j=1}^d w_j g_{ij}^0(k)}\right)= \sum_{j=1}^d w_j g_{ij}^1(k) .
\end{equation}

Thus, the poverty degree of a multidimensionally poor individual is the weighted sum of the censored positional depth scores, with weights summing to one across all indicators; equivalently, it is the individual's intensity multiplied by the average depth of the indicators in which the individual is deprived. Higher values indicate greater poverty levels, with zero indicating non-poverty, one indicating total poverty, and values between zero and one indicating intermediate poverty situations. Therefore, the measure follows the fuzzy logic, and $P_i$ can be interpreted as a membership function.

\subsection{Aggregated multidimensional poverty}

The incidence of poverty, or poverty headcount ratio, is widely used in empirical research and official national statistics. It is simple, easy to understand, and familiar to policymakers and the public. Therefore, it is important to include the incidence as a component of the index while adjusting for its limitations.

The incidence of multidimensional poverty is the percentage of people who are identified as poor given a poverty line, $k$. It can be represented as follows:

\begin{equation} \label{eq:5}
H = \frac{1}{n}\sum_{i=1}^n \rho_i(k) = \frac{q}{n}, 
\end{equation}
where $\rho_i(k)$ is the poverty identification function, $n$ total population, and $q$ number of poor. $H$ therefore measures the proportion of people who are identified as multidimensionally poor.

As one can see in Equation \ref{eq:5}, $H$ is obtained only as an aggregated measure, as it is not possible to have individual levels. That is why it may be preferable to use $H$ as a separate measure from the poverty degrees, which is defined as an individual measure, and adjust it to better represent the different aspects of poverty\footnote{For simplicity, users can also employ $H$ as a standalone poverty indicator and interpret it in conjunction with the poverty degree, revealing that a share $H$ of the population is experiencing poverty but at varying degrees.}. 

I then define the aggregated multidimensional poverty, or adjusted positional poverty gap, as the product of the three censored averages:
\begin{equation}
P = H \cdot A \cdot S\;=\; \frac{1}{n}\sum_{i=1}^{n} P_i ,
\end{equation}
where $A$ is the average intensity among the poor and $S$ is the  average positional poverty gap among the poor, as defined in Sections \ref{section: intensity} and \ref{section: distributional_poverty_gap}. Hence the adjusted positional poverty gap index $P$ equals the population mean of individual poverty degrees, where the weighting by deprivation episodes is what makes the product identity exact. 

The positional depth scores $s_{ij}$ can be raised to a power $\alpha \geq 1$, with $\alpha > 1$ placing greater weight on the most deeply deprived. Writing the generalized censored score as $g^\alpha_{ij}(k) = g^0_{ij}(k)\, s_{ij}^\alpha$, this defines a generalized index $P_\alpha = \frac{1}{n}\sum_i\sum_j w_j\, g^\alpha_{ij}(k)$, which is analogous to the AF generalized family
$M_\alpha = \frac{1}{n}\allowbreak\sum_i \sum_j w_j\, \big(g^{AF}_{ij}(k)\big)^{\alpha}$, obtained by raising the AF gap to the same power, where $g^{AF}_{ij}(k) = \rho_i(k)\, g^{0}_{ij}\, g^{AF}_{ij}$ is the censored AF gap. When $\alpha = 1$, $P$ is the positional poverty gap index,
analogue of the AF adjusted poverty gap $M_1 = H \cdot A \cdot G$, 
where $G$ is the normalized poverty gap \citep[see][]{AlkireEtAl2015}, the average AF poverty gap per deprivation episode among the poor, $G = \dfrac{\sum_{i=1}^{n}\sum_{j=1}^d w_j\, g^{AF}_{ij}(k)}{\sum_{i=1}^{n}\sum_{j=1}^d w_j g_{ij}^0(k)}$, defined analogously to $S$. Its individual counterpart, $G_i = \dfrac{\sum_{j=1}^d w_j\, g^{AF}_{ij}(k)}{\sum_{j=1}^d w_j g_{ij}^0(k)}$, is the cardinal counterpart of the individual positional poverty gap $S_i$.
Note that when all indicators are binary, $P$ recovers the AF's $M_0$.

\section{Characterization and properties}
 \label{section: properties}

\subsection{Characterization of the positional depth score}
\label{section: rd_foundation}

Section~\ref{section: distributional_poverty_gap} introduced the positional depth score and showed that, once the score is proportional to the share of individuals better off and the most deprived individual is required to score one, the normalizer is uniquely determined. This subsection derives the proportionality itself, building on the relative-deprivation literature introduced in Section~\ref{section: background}. I show that replacing the cardinal scale conditions of \citet{EbertMoyes2000} with the ordinal invariance requirement of Section~\ref{section: 3.2} turns Yitzhaki's index \citep{Yitzhaki1979} into the positional depth score. The score is therefore not one selection among admissible rank transformations: it is the unique normalized deprivation score consistent with this comparison structure and with the informational content of ordinal data.

The characterization rests on one normative premise, applied only after identification: an individual is deprived when $x_{ij} < z_j$. For ordinal indicators the absolute distance to the threshold is not defined, so depth must be assessed relationally; the premise adopted is that, among the deprived, a position below a larger share of the reference population is the worse condition, since more of that population attains what the individual does not. This judgment concerns positional disadvantage only. It does not imply that a rare deprivation is intrinsically more harmful than a widespread one. The conditions below give the judgment its formal content.

Fix an indicator $j$ and let $\mathcal{R}$ denote the reference population, the finite set of individuals whose achievements generate the reference distribution $F_j$: the current sample under in-sample CDFs, a baseline or pooled sample under anchored CDFs (Section~\ref{section: comparability_time_space}). Following \citet{EbertMoyes2000}, the primitive is a nonnegative deprivation function $\delta_{ij}(X_j,B)$, defined for any achievement profile, the deprivation of individual $i$ in indicator $j$ with respect to a group $B\subseteq \mathcal{R}$, where $X_j$, the $j$-th column of the achievement matrix, collects the achievements in indicator $j$. Deprivation is understood here as a condition rather than a perception. The basic element is that another individual ranks above $i$, and who ranks above whom is exactly what an ordinal variable records. The positional depth score of a deprived individual is the deprivation with respect to the full reference population, rescaled so that the most deprived member of $\mathcal{R}$ receives a score of one, which is the normalization of Section~\ref{section: distributional_poverty_gap}. I impose four conditions on the deprivation function.

\emph{Pairwise comparison.} For any single individual $h$, the deprivation of $i$ with respect to $h$ depends only on the two achievements involved, through a function common to all individuals and comparators: $\delta_{ij}(X_j,\{h\}) = f(x_{ij},x_{hj})$. Each pairwise comparison is self-contained, and neither the identity of $i$ nor that of the comparator matters. This condition merges the independence and anonymity conditions of \citet{EbertMoyes2000}.

\emph{Focus.} For any $h\in B$ with $x_{hj}\le x_{ij}$, $\delta_{ij}(X_j,B)=\delta_{ij}(X_j,B\setminus\{h\})$. Deprivation is generated only by individuals with strictly higher achievements; those at the same position or below contribute nothing.

\emph{Additive decomposition.} For disjoint $B_1,B_2\subseteq \mathcal{R}$, $\delta_{ij}(X_j,B_1\cup B_2)=\delta_{ij}(X_j,B_1)+\delta_{ij}(X_j,B_2)$. Comparisons aggregate without complementarities: the deprivation added by one further better-off person does not depend on how many others already rank above $i$.

\emph{Ordinal invariance.} For every strictly increasing transformation $\varphi$ of the real line, $\delta_{ij}(\varphi(X_j),B)=\delta_{ij}(X_j,B)$, where $\varphi(X_j)$ applies $\varphi$ to every achievement; achievements are coded as real numbers, of which only the ordering is meaningful. This is the admissible-transformation requirement of Section~\ref{section: 3.2}, applied to the deprivation function itself.

\begin{theorem}\label{thm:rd_foundation}
Fix an indicator $j$ with reference distribution $F_j$ such that $F_j(\min_i x_{ij})<1$. If the deprivation function is not identically zero and satisfies pairwise comparison, focus, additive decomposition, and ordinal invariance, then the deprivation of individual $i$ with respect to the reference population is proportional to the share of that population with achievements strictly above $x_{ij}$, that is, to $1-F_j(x_{ij})$. Consequently, under the normalization of Section~\ref{section: distributional_poverty_gap}, the score of any achievement not below the reference minimum is uniquely
\[
s_{ij}=\frac{1-F_j(x_{ij})}{1-F_j(\min_i x_{ij})}.
\]
\end{theorem}

\textbf{Proof.} The proof is in Appendix~\ref{appendix: B}, together with a remark showing that the conditions are independent.
\bigskip

The theorem also locates the exact difference from Yitzhaki's index, which the same comparison structure delivers under cardinal scale conditions \citep{EbertMoyes2000}. That index factors into the share of individuals ranking above $i$ and the average size of the corresponding shortfalls; the second component is not preserved under ordinal invariance, so the positional depth score is not the same quantity, but what remains of it when only the event of being outranked, and not its size, is meaningful. Under the four conditions, the deprivation value can depend only on how much of the reference population ranks above the individual, so ordinal depth is measured in population shares rather than in units of the underlying variable. The substantive condition is additive decomposition, which excludes power transformations of the score, $s_{ij}^{\gamma}$ with $\gamma\neq 1$. Rejecting additivity is a coherent normative position, and the generalized family $P_\alpha$ of Section~\ref{section: overall_poverty} shows where such a preference belongs: the theorem concerns the measurement of individual deprivation, while the exponent $\alpha$ expresses the evaluator's aversion to deep deprivation at the aggregation stage, as the FGT exponent operates on a linear individual gap.

The theorem holds for any choice of reference distribution, in-sample or anchored. Its scope conditions are those already stated: when $F_j(\min_i x_{ij})=1$ the convention $s_{ij}=1$ applies as the continuous limit of the formula, and under anchored CDFs achievements below the reference support are scored at one as defined in the next subsection. For cardinal indicators, imposing ordinal invariance is a choice to use only the ordinal content of the variable; it is this choice that gives the scores a common unit across indicators, while magnitude-based measures such as the AF gap remain admissible where their scale requirements are met.

\subsection{Axiomatic properties}
Axiomatic approaches to multidimensional poverty specify properties designed to guarantee well-behaved poverty indices. Most properties are temporal, as they specify how the index should respond when recomputed for a new time point in a panel. In this section, I discuss and prove whether standard axioms used in the multidimensional poverty literature \citep{AlkireFoster2011, AlkireEtAl2015} apply to the proposed method. Here, I focus on the axioms for which the justification differs from that in the AF framework due to the positional poverty gap component. These axioms are: deprivation and poverty focus, monotonicity, subgroup consistency and decomposability, and weak transfer. The remaining properties (symmetry, replication invariance, and weak rearrangement) are briefly discussed, as their proofs are direct and follow the same argument as for the AF measures. 
Moreover, to analyze the innovative part of the index, I focus on $P$ (or, more generally, $P_\alpha$ with $\alpha \geq 1$), distinguishing between individual or aggregated levels when necessary. This focus is due to the individual and aggregated intensity here being the same as in the AF framework, and their properties are discussed elsewhere (see \citet{AlkireEtAl2015}). 

The distinction between a fixed and a sample-dependent reference distribution is what drives the axiomatic behavior of the index. A related caution appears in the literature on data-driven index components: \citet{DuttaEtAl2021} show that endogenous weighting systems, including frequency-based weights that depend on how often deprivation occurs in each indicator, can violate monotonicity and subgroup consistency, because one individual's change then alters other individuals' contributions to the aggregate. The positional poverty gap does not use endogenous weights, but an in-sample CDF introduces an analogous sample-dependent element into the depth score, with comparable consequences. This subsection shows that anchoring the reference distribution removes this dependence.

I consider two possible implementations: (i) \emph{anchored CDFs}, where each $F_j$ is estimated once either on a baseline sample (then used to score subsequent years) or on a pooled sample containing all years under analysis (one-time estimation for the full set), and is then held fixed for evaluation, with any achievement below the reference support in subsequent samples scored at one, as at the reference minimum; (ii) \emph{in-sample CDFs}, where $F_j$ is computed from the current sample, without a baseline reference, which means that the CDF is recalculated at each time period. Table \ref{table:tab1} summarizes the properties and whether they hold or not, considering the two possible implementations.

\begin{table}[ht]
\centering
\caption{Axioms summary}
\setlength{\tabcolsep}{3pt}
\begin{tabular}{lcc}
\toprule
Property & Anchored CDF & In-sample CDF \\
\midrule
Symmetry, replication, bounds        & \cmark & \cmark \\
Ordinal invariance (admissible transf.) & \cmark & \cmark \\
Deprivation focus                     & \cmark & \cmark \\
Poverty focus                         & \cmark & $\diamond^1$ \\
Own-person monotonicity               & \cmark & \cmark \\
Aggregate monotonicity                & \cmark & \xmark \\
Dimensional monotonicity              & \cmark & \xmark \\
Decomposability \& subgroup consistency & \cmark & \xmark \\
Weak rearrangement                    & \cmark & \cmark  \\
Weak transfer                         & \xmark & \xmark \\
\bottomrule
\end{tabular}
\label{table:tab1}%
{\notesize Notes: \cmark\ satisfied; $\diamond$ conditional; \xmark\ fails. 1. Satisfied only when the poverty identification is by union ($k = \min_j w_j$).\par}
\end{table}

In what follows, I present a formal and an intuitive presentation of the properties. The proofs are in Appendix~\ref{appendix: B}.

\emph{Deprivation Focus.} Let $X'$ be obtained from $X$ by changing a single cell $(i,\theta)$
such that $x'_{i\theta}>x_{i\theta}\ge z_\theta$ (an improvement in a non-deprived indicator), and
$x'_{pj}=x_{pj}$ for all $(p,j)\neq(i,\theta)$. A poverty measure $P$ satisfies
deprivation focus if
\[
P(X';z,k)=P(X;z,k).
\]

This axiom guarantees that non-deprived indicators do not influence the index. $P$ satisfies this axiom. This result follows from censoring the values of the positional depth score for nondeprived indicators. 

\textit{Poverty Focus.} \,
Let $X'$ be obtained from $X$ by changing the achievements of a single person $i$
so that $\rho_i(k)=\rho'_i(k)=0$ (the person remains non-poor), with $x'_{pj}=x_{pj}$ for all
$(p,j)\neq(i,\cdot)$. A poverty measure $P$ satisfies \emph{Poverty Focus} if
\[
P(X';z,k)=P(X;z,k).
\]

This axiom states that changes in the achievements of the \emph{non-poor} do not affect measured poverty.
$P$ satisfies this property when the CDFs are \emph{anchored} and also under \emph{in-sample} CDFs if the identification is by \emph{union} ($k=\min_j w_j$), because non-poor then have no deprived cells. 
The axiom may \emph{fail} with \emph{in-sample} CDFs and intermediate $k$, but only in a specific case: a change by a non-poor person (who has at least one deprivation) to one of their deprived indicators ($x_{i\theta}<z_\theta$) can alter the empirical CDF for indicator $\theta$ and thereby the positional depth scores of other individuals, so that $P(X';z,k)$ may differ from $P(X;z,k)$.

\textit{Monotonicity.} \, Let $X'$ be obtained from $X$ by worsening person $i$ in indicator $\theta$: either (a) $x'_{i\theta}<x_{i\theta}<z_\theta$ (still deprived), or
(b) $x'_{i\theta}<z_\theta\le x_{i\theta}$ (crosses into deprivation). Monotonicity requires $P(X';z,k)\ge P(X;z,k)$, with strict inequality whenever the person is (or becomes) poor, the worsening is strict in a deprived indicator, the score $s_{i\theta}$ has not already attained its maximum of one, and the reference distribution places positive mass in the interval $(x'_{i\theta}, x_{i\theta}]$, where under in-sample CDFs this mass must come from individuals other than $i$ (in particular, whenever the person moves past at least one category occupied by others).

For own-person monotonicity, $P_i$ does not decrease when a deprivation worsens. The positional depth score depends only on the individual's position in the reference distribution: if $x_{ij}$ falls, the share of individuals with higher achievements can only remain constant or rise, so $s_{ij}$, and hence $P_i$, cannot decrease, and the increase is strict under the positive-mass condition stated above.

For aggregated monotonicity, violation or not depends on whether one uses \textit{anchored CDFs} or \textit{in-sample CDFs}. For the latter, monotonicity need not hold. Because $F_j$ is a relative measure, when it is recomputed from the sample, a change in $x_{i\theta}$ shifts all positional depth scores in column $\theta$. For a worsening within deprivation that leaves the set of the poor unchanged, the aggregate change is:
\[
\Delta P \;=\; \frac{1}{n}\,w_\theta\Big[(s'_{i\theta}-s_{i\theta})\rho_i(k)g^0_{i\theta} \;+\; \sum_{p\neq i}\big(s'_{p\theta}-s_{p\theta}\big)\rho_p(k)g^0_{p\theta}\Big].
\]
The second (externality) term can be negative and dominate the first; hence $\Delta P$ can be $<0$ even though individual $i$ worsens in a deprived indicator. Therefore aggregate monotonicity \emph{need not hold} with in-sample CDFs.%
\footnote{With in-sample CDFs, aggregate monotonicity can fail via two channels, even if the poor set is fixed. 
(i) \emph{Denominator (minimum) effect:} if the share at an unchanged minimum $m_\theta$ changes by one observation, the common normalizer
$D_\theta^{\mathrm{IS}}=1-F_\theta^{\mathrm{IS}}(m_\theta)$ shifts by $\Delta D_\theta^{\mathrm{IS}}=\pm 1/n$ (use the weighted analogue with survey weights); if the worsening establishes a new minimum, $D_\theta^{\mathrm{IS}}$ jumps to one minus the mass at the new minimum, a shift that can be much larger. Either way, all positional depth scores in column $\theta$ are rescaled:
$\Delta s^{\mathrm{IS}}_{p\theta}=-(1-F_\theta^{\mathrm{IS}}(x_{p\theta}))\,\Delta D_\theta^{\mathrm{IS}}/(D_\theta^{\mathrm{IS}}\,D_\theta^{\mathrm{IS}\prime})$, where $D_\theta^{\mathrm{IS}\prime}$ is the new normalizer, even if ranks are unchanged.
(ii) \emph{Peer-redistribution effect:} with $m_\theta$ fixed, moving one observation from $a$ to $b>a$ (both below $z_\theta$) changes the empirical CDF by
$\Delta F_\theta^{\mathrm{IS}}(x)=-1/n$ for peers with $x\in[a,b)$ (and $+1/n$ if $b<a$), shifting their numerators and hence
$\Delta s^{\mathrm{IS}}_{p\theta}=-\Delta F_\theta^{\mathrm{IS}}(x_{p\theta})/D_\theta^{\mathrm{IS}}$. Either effect can outweigh the worsened person’s increase, so the aggregate change can be negative. For a concrete example, consider a single indicator with $z_\theta=5$ and eleven individuals, five with achievement $1$, four with $2$, one with $3$, and one non-deprived with $10$; under union identification $P=6.5/11\approx0.59$. If the individual at $3$ worsens to $0$, the normalizer rises from $6/11$ to $10/11$, every other deprived score deflates, and $P$ falls to $3.9/11\approx0.35$.}

For the \textit{anchored CDFs} monotonicity will hold, as changes in $x_{i\theta}$  will not shift any other cells’ depth terms.

\textit{Subgroup Consistency and Decomposability.} \,
Let the population be partitioned into $L\ge2$ mutually exclusive and collectively exhaustive subgroups
$\ell_1,\ldots,\ell_L$ with sizes $n^\ell$ (so $\sum_{\ell=1}^L n^\ell=n$), and let $X^\ell$ denote the
achievement matrix restricted to subgroup $\ell$. Fix cutoffs $(z,k)$, weights $w$, and evaluate the
total population and all subgroups using the \emph{same} reference CDFs (anchored). Then:

\textit{Subgroup Consistency.}
If $X'$ is obtained from $X$ by changing outcomes \emph{only} in subgroup $\ell'$ so that
$P(X'^{\,\ell'};z,k)<P(X^{\ell'};z,k)$ and $P(X'^{\,\ell};z,k)=P(X^{\ell};z,k)$ for all $\ell\neq\ell'$,
while subgroup sizes remain the same, then
\[
P(X';z,k)\;<\;P(X;z,k),
\]
and the inequality reverses if subgroup $\ell'$ worsens.
\bigskip

\textit{Population Subgroup Decomposability.}
The overall index equals the \emph{population-weighted mean} of subgroup indices, where each subgroup contributes
in proportion to its population share:
\[
P(X;z,k)\;=\;\sum_{\ell=1}^{L}\Big(\frac{n^\ell}{n}\Big)\,P(X^\ell;z,k).
\]

This identity holds whether one computes the subgroups first or the total first, provided all are
evaluated against the same reference CDFs and fixed $(z,k,w)$.

\textit{Weak transfer.} An equalizing swap within one indicator, taking a little achievement from a less-deprived poor person and giving it to a more-deprived poor person, keeping their average fixed, should not raise measured poverty. However, the proposed measures do not satisfy this property, because they use a rank–based penalty $s_{ij} \propto 1-F_j(x_{ij})$. Its curvature depends on the shape of $F_j$: if the density rises over the deprived range, $1-F_j(\cdot)$ is \emph{concave}, so an equalizing (Pigou–Dalton) transfer between two poor in indicator $j$ can \emph{increase} the sum of penalties. Hence, unlike AF gaps (which are convex for $\alpha\!\ge\!1$), weak transfer is not distribution-free for CDF–based depth; with in-sample CDFs it can fail even more often due to denominator (minimum) rescaling and peer-redistribution effects\footnote{For a concrete example under anchored CDFs, consider an indicator with four equally spaced values $1,2,3,4$ carrying population shares $0.1$, $0.1$, $0.4$, and $0.4$, with $z_j=4$. The anchored scores of the deprived values are $s=1$, $0.889$, and $0.444$. A progressive transfer moving one poor person from $3$ to $2$ and another from $1$ to $2$ raises the sum of their scores from $1.444$ to $1.778$, so measured poverty increases. Since this argument applies to the cardinal subdomain, where transfers are well defined, the axiom fails in general.
}. 
The weak transfer in the AF-framework holds only when $\alpha\geq1$, which is rarely used. In practice, therefore, most studies do not comply with this property and do not rely on it. 

The methodology also satisfies the following properties. 
\medskip

\emph{Symmetry.}\, Relabeling individuals or indicators (with weights aligned) leaves $P$ unchanged. 

\medskip

\emph{Replication invariance.}\, Duplicating the population does not change $P$ since it is a (weighted) mean of cell contributions. 

\medskip

\emph{Bounds/normalization.}\, $P\in[0,1]$, equals $0$ when no one is poor, and reaches $1$ only at maximal deprivation depth and breadth. The upper bound follows from the normalized score $s_{ij}\le 1$: under in-sample CDFs the normalizer is evaluated at the sample minimum, and under anchored CDFs achievements below the reference support are scored at one, as defined above.

\medskip

\emph{Ordinal invariance.}\, Monotone recodings of ordinal indicators do not affect $P$ because $s_{ij}$ depends only on ranks via $F_j$. 

\medskip

\emph{Weak rearrangement.}\, Under an association-decreasing rearrangement among the poor, in which two poor individuals exchange achievements in a subset of indicators while both remain poor, $P$ does not increase. In fact $P$ is unchanged: each indicator's multiset of values is preserved by the exchange, so every score $s_{ij}$ attached to a given value is unchanged (under anchored CDFs because the reference is fixed, and under in-sample CDFs because the empirical distributions are unaffected by the swap), and the exchanged censored cells contribute the same total. The measure therefore satisfies weak rearrangement with equality, as does $M_0$.

\medskip

\begin{theorem}\label{thm:theorem} 
Given weights $w$, cutoffs $(z,k)$, and $P_\alpha$ with $\alpha \geq 1$, the methodology with positional depth scores satisfies:
\emph{symmetry}, \emph{replication invariance}, \emph{bounds/normalization}, \emph{ordinal invariance}, and \emph{deprivation focus} (anchored or in-sample); 
\emph{poverty focus} under \emph{anchored} CDFs for any $k$, and under \emph{in-sample} CDFs only under \emph{union} identification; 
\emph{individual-level monotonicity} (anchored or in-sample); 
\emph{aggregated-level} and \emph{dimensional monotonicity}  under \emph{anchored} CDFs;
\emph{subgroup decomposability} and \emph{subgroup consistency} under \emph{anchored} CDFs; and
\emph{weak rearrangement} (anchored or in-sample).%
\end{theorem}

\textbf{Proof.} The proof for Theorem \ref{thm:theorem} is in Appendix~\ref{appendix: B}.
\bigskip

\subsection{Correspondence with the Alkire–Foster poverty gap}
\label{section: correspondence_AFPG_PPG} 

The positional poverty gap and the Alkire--Foster poverty gap measure different concepts. Although both reflect the depth of poverty, they do so from distinct perspectives: the AF poverty gap measures the absolute shortfall from each deprivation threshold, while the positional poverty gap measures relative position within the achievement distribution. Thus, the two scores need not coincide in magnitude. Within a
dimension, however, they rank deprived observations in the same order when the
reference CDF distinguishes the observed deprived values. This condition holds automatically under in-sample CDFs, where every observed deprived value carries positive mass; under an anchored reference it can fail if two deprived values observed in the current sample are not separated by the reference distribution, in which case the two scores tie.

\begin{proposition}\label{prop:within-dim}
	Fix a dimension $j$. Suppose that $z_j>0$, $1-F_j(\min_i x_{ij})>0$, and that, for any
	two observed deprived values $u$ and $v$ in dimension $j$,
	$u<v<z_j  \Rightarrow F_j(u)<F_j(v).$
	Then, for any two individuals $a$ and $b$ deprived in dimension $j$,
	\[
	x_{aj}\le x_{bj}
	\quad\Longleftrightarrow\quad
	g^{AF}_{aj}\ge g^{AF}_{bj}
	\quad\Longleftrightarrow\quad
	s_{aj}\ge s_{bj},
	\]
	with strict inequalities whenever $x_{aj}<x_{bj}$. Hence the positional depth score and
	the Alkire--Foster normalized gap induce the same within-dimension ranking of deprived
	observations.
\end{proposition}

\textbf{Proof.} The proof is given in Appendix~\ref{appendix: B}. It uses only the fact that the AF gap is strictly decreasing in achievement over the deprived range; the proposition therefore extends to any cardinal deprivation score with this property, including shortfalls of interval-scale indicators normalized by a fixed band, as used for nutrition in Section~\ref{section: illustration_mics}.
\bigskip 

Proposition~\ref{prop:within-dim} implies that any divergence between the rankings of the
individual positional poverty gap $S_i$ and the Alkire--Foster poverty gap $G_i$ does not come from
within-dimension reversals. Rather, it arises from the aggregation of
dimension-specific scores that use different reference points: the deprivation threshold
for $g^{AF}_{ij}$ and the achievement distribution for $s_{ij}$. If all poor individuals
are deprived in a single common dimension, the two measures rank them identically under
the conditions stated above.

How large this cross-dimensional divergence is depends on how the two reference points relate within each dimension. Under the Alkire--Foster poverty gap, a deprivation scores more when the achievement lies far below the deprivation threshold. Under the positional poverty gap, the same deprivation scores more when it leaves the individual behind a larger share of the population, regardless of its distance from the threshold. The two measures therefore may diverge in two opposite situations. A deprivation that is shallow in absolute terms but rare in the population, so that nearly everyone else is better off, receives a small AF gap but a high positional depth score. Conversely, a deprivation that is deep in absolute terms but widespread, so that many others share similarly low achievements, receives a large AF gap but a more moderate positional depth score. Individuals whose deprivations are concentrated in dimensions of the first type are ranked as poorer by the positional poverty gap than by the AF gap, and the reverse holds for the second type. Divergence between the two rankings is thus largest when the poor differ systematically in the type of dimension in which they are deprived. When individuals are instead deprived in several dimensions of opposite type, the corresponding differences partly cancel in the weighted average, so the two measures tend to rank them more similarly.

The two scores also span different ranges among the deprived. For non-negative
achievements, the Alkire--Foster gap takes values in $(0,1]$. The positional depth score instead
has a floor set by the prevalence of deprivation. Let $x_j^{-}$ denote the highest
deprived value observed in dimension $j$, so that $F_j(x_j^{-})$ equals the uncensored
headcount ratio of indicator $j$, denoted $h_j$. Since $s_{ij}$ is non-increasing in
$x_{ij}$, every deprived individual satisfies
\begin{equation}\label{eq:floor}
	s_{ij} \;\ge\; \frac{1-h_j}{1-F_j(\min_i x_{ij})} \;\ge\; 1-h_j .
\end{equation}
Every deprived individual thus scores at least the share of the population not
deprived in $j$. The lower range of the AF gap is set by the metric of the indicator,
while that of the positional depth score is set by how exceptional deprivation is in
the reference population.

 Subsection~\ref{section: illustration_mics} illustrates these results empirically.

\section{Illustrative application}
 \label{section: illustrations}
This section presents two applications of the method. The first draws on Brazilian administrative data and reports results across poverty lines and population subgroups, illustrating the potential added value of the positional poverty gap measure.
 
The second application uses data from 23 countries and is restricted to cardinal variables. It reports rank correlations and compares the rank distributions generated by the positional poverty gap and the AF poverty gap. The analysis examines the empirical relationship between the two measures and assesses whether, despite the high correlations observed across the country samples, they provide complementary information.
 
 \subsection{Brazil's Cadastro Único data}
This illustration applies the positional poverty gap to a sample from Brazil's Single Registry for Social Programmes (\textit{Cadastro Único}) 2018 \citep{brasil_cadunico_2018}, the administrative data used to assess eligibility for the Bolsa Família conditional cash transfer and other social programs. The sample consists of 4,807,996 households (12,852,599 individuals), representative of Cadastro Único registrants. With the information available in this database, I created the poverty index using eight indicators grouped into three dimensions: (1) education, with an adult-attainment indicator and a child age-for-grade indicator; (2) housing, including inadequate housing materials and crowding (persons per bedroom); and (3) basic services, with drinking water access, type of sanitation facility, garbage disposal, and access to electric lighting. The three dimensions are equally weighted, as are the indicators within each dimension. Table \ref{table:C_defs} details the indicator definitions and thresholds. All measures are computed at the household level using the survey weights.

 Because not all indicators are cardinal, it is not possible to calculate the AF poverty gap, so Table~\ref{tab:poverty_by_cutoff_CAD} shows the results for $H$, $H \cdot A$, and $P = H \cdot A \cdot S$. The results are presented for three poverty cutoffs: deprived in at least one indicator ($k=1/12$), which is equivalent to the union approach; deprived in the equivalent of one full dimension ($k=1/3$), the reference cutoff used throughout the rest of this illustration; and deprived in the equivalent of two of the three dimensions ($k=2/3$).

Table~\ref{tab:poverty_by_cutoff_CAD} shows that for all measures the values decrease as the poverty cutoff increases. When $k=1/3$, the share of multidimensionally poor households equals $H=31.7\%$, and the intensity equals $A=0.44$. The positional poverty gap equals $S=0.84$, meaning that, among the indicators in which they are deprived, poor households are on average 84\% as deeply deprived as the most deprived household observed in that indicator. These last two indices give the $H \cdot A$ and $P = H \cdot A \cdot S$ reported in the table.

\begin{table}[H]
      \centering
      \caption{Results by poverty cutoff ($k$), Brazil, Cadastro Único 2018}
      \begin{tabular}{lccc}
        \toprule
        Measure & $k=1/12$ (union) & $k=1/3$ & $k=2/3$ \\
            \midrule
        $H$                      & 0.810 & 0.317 & 0.032 \\
        $H \cdot A$              & 0.227 & 0.140 & 0.023 \\
        $P = H \cdot A \cdot S$  & 0.187 & 0.118 & 0.020 \\
        \bottomrule
      \end{tabular}
      \begin{tablenotes}
        \notesize
        \item \textit{Notes:} Survey-weighted. Population poor at each score cutoff $k$, using eight indicators grouped into three equally weighted dimensions (education, housing, services). $k=1/12$ is the union cutoff.
      \end{tablenotes}
      \label{tab:poverty_by_cutoff_CAD}
    \end{table}

     Figure~\ref{fig:cad_P_decomp} decomposes $P$ at the reference cutoff into the contribution of each of the eight indicators. Adult education alone accounts for 27.5\% of $P$, followed by crowding (17.4\%); sanitation, water, and the child education gap each contribute 11--13\%; housing materials and garbage contribute smaller, comparable shares (7--10\%); and electricity, the least prevalent deprivation in these data, contributes only 1.8\%.
     
\begin{figure}[H]
    \centering
    \includegraphics[width=0.75\textwidth]{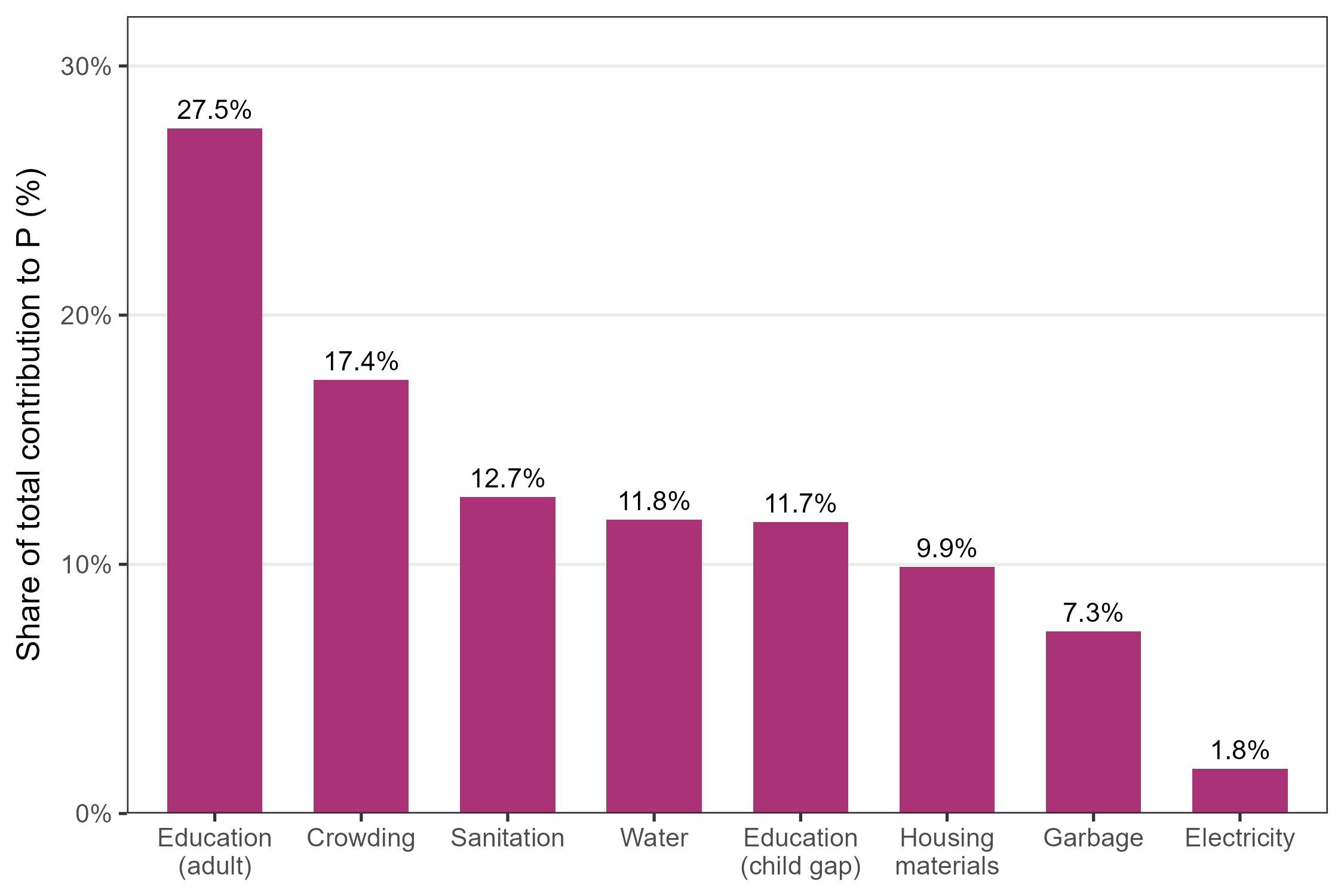}
    \caption{Decomposition of the positional poverty gap index $P$ by indicator, Brazil, Cadastro Único 2018.}
    \label{fig:cad_P_decomp}
    \vspace{4pt}
    \begin{minipage}{\textwidth}
    \noindent {\notesize \textit{Notes:} Survey-weighted. Bars show the percentage contribution of each indicator to $P$ at the reference cutoff $k=1/3$.\par}
    \end{minipage}
\end{figure}

Figure~\ref{fig:cad_dominance} plots the positional poverty gap index $P$ across poverty cutoffs $k$ for population subgroups. Poverty dominance is robust over the whole range of $k$: the North shows the highest values, followed by the Northeast, while the Southeast is the lowest; and households headed by Indigenous persons show the highest values, followed by those headed by Black or Pardo persons, and then by White or Asian persons. The curves converge as $k$ rises, so that at stricter cutoffs the subgroups differ less.

\begin{figure}[H]
    \centering
    \begin{subfigure}[t]{0.48\linewidth}
      \centering
      \includegraphics[width=\linewidth]{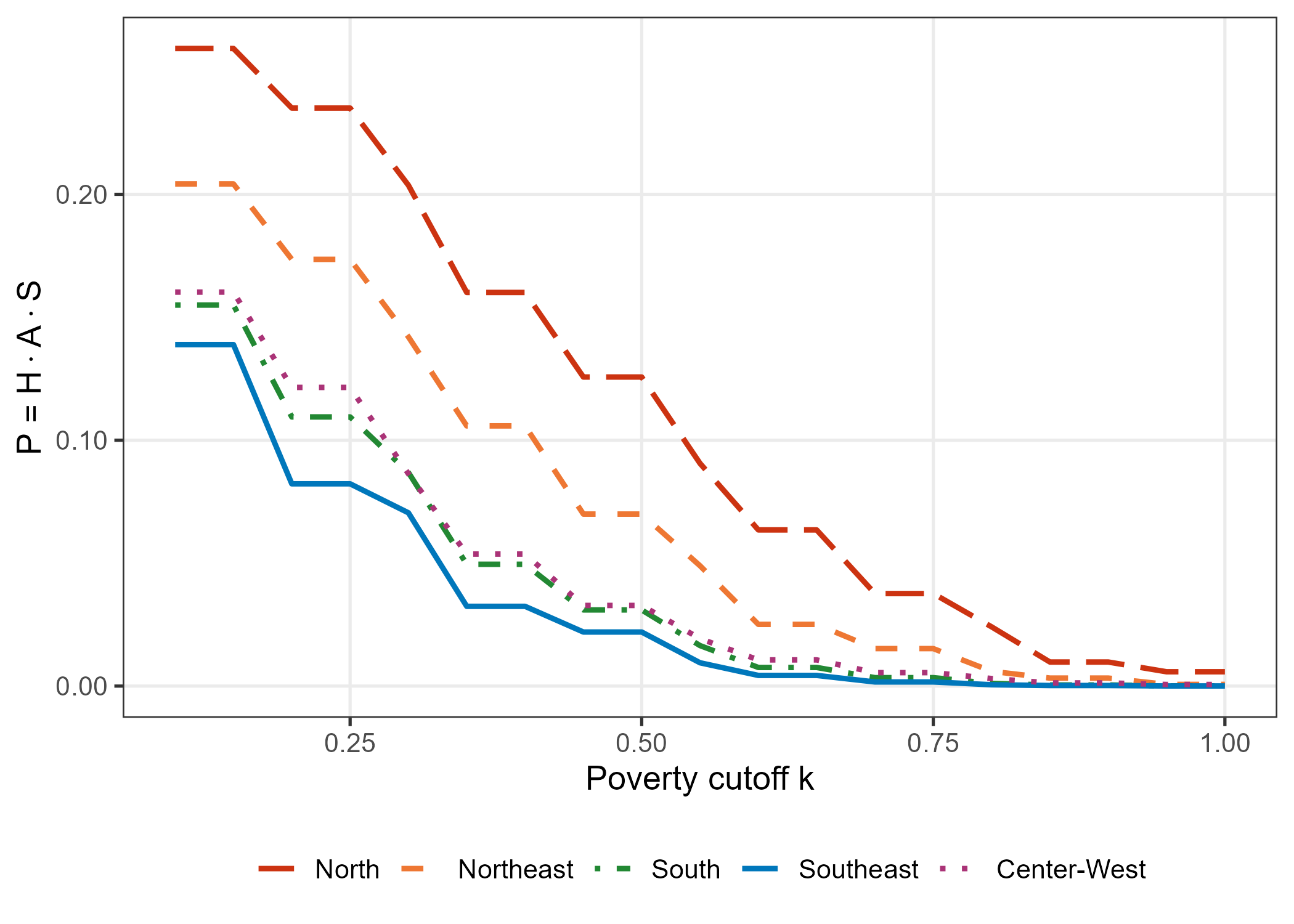}
      \caption{By macro-region}
      \label{fig:cad_dom_region}
    \end{subfigure}%
    \hspace{0.04\linewidth}%
    \begin{subfigure}[t]{0.48\linewidth}
      \centering
      \includegraphics[width=\linewidth]{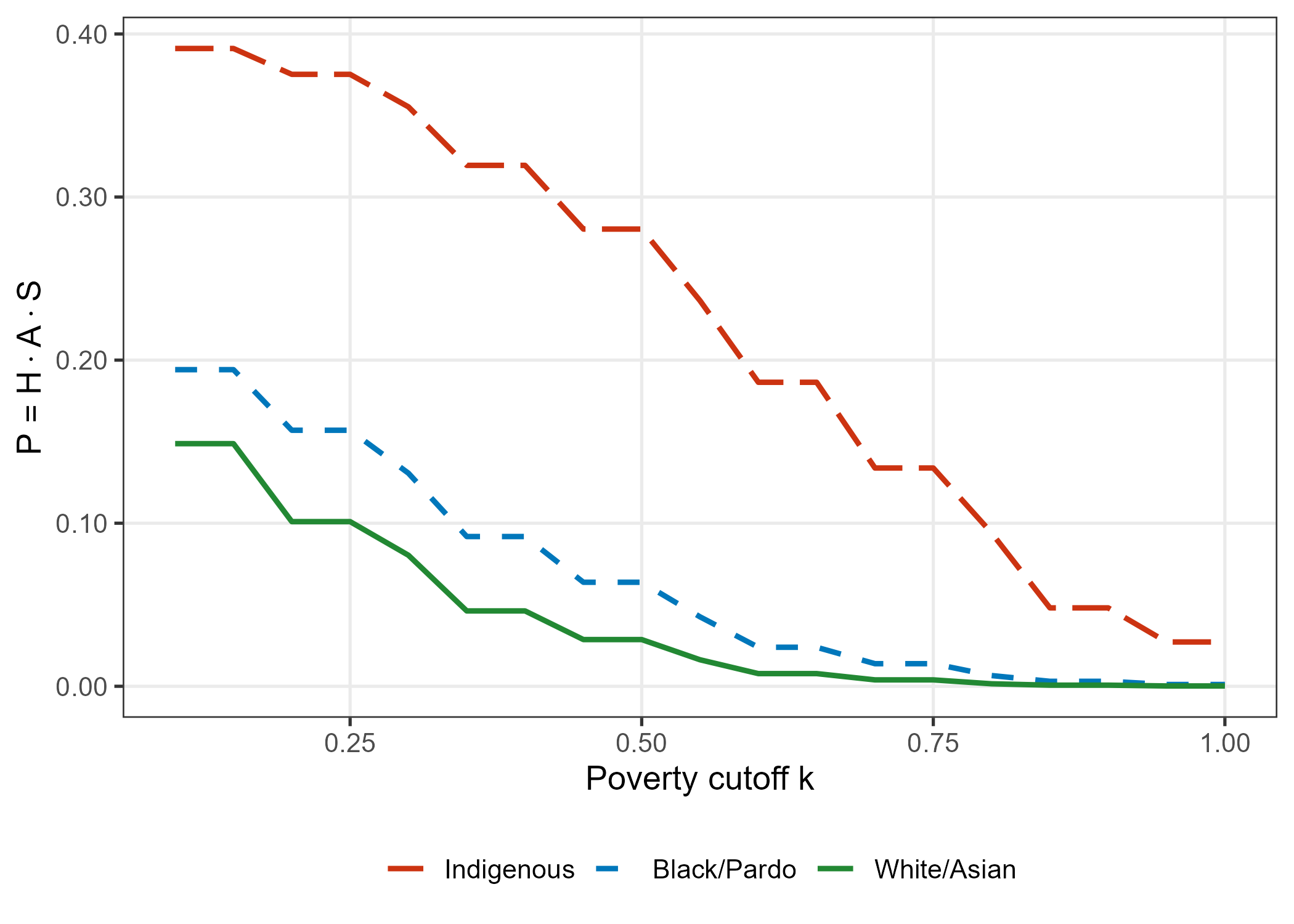}
      \caption{By race/ethnicity of household head}
      \label{fig:cad_dom_race}
    \end{subfigure}
    \caption{Poverty dominance of the positional poverty gap index ($P = H \cdot A \cdot S$) by poverty cutoff $k$, Brazil, Cadastro Único 2018.}
    \label{fig:cad_dominance}
    \vspace{4pt}
    \begin{minipage}{\textwidth}
    \noindent {\notesize \textit{Notes:} Survey-weighted. Each curve plots the positional poverty gap index $P$ against the score cutoff $k$ for a subgroup: macro-region in panel (a), race/ethnicity of the household head in panel (b).\par}
    \end{minipage}
\end{figure}

Figure~\ref{fig:cad_scatter_Si_Ai}(a) shows how the positional poverty gap varies within each intensity group. Even among households deprived in the fewest dimensions, $S_i$ ranges widely, up to its maximum: a household with only one or two deprivations can still be very deeply deprived. This variation is invisible to measures based on incidence and intensity alone, which would assign such a household the same value as any other with the same number of deprivations. As intensity rises, the positional gap tends to rise as well, though the interquartile range narrows only modestly.

 Figure~\ref{fig:cad_scatter_Si_Ai}(b) repeats this exercise among households actually receiving Bolsa Família, grouped by income per capita quintile. Depth is only mildly related to income among the program's beneficiaries: the median declines slightly as income rises, but every quintile, including the highest-income one, spans essentially the same range, from the shallowest to the deepest observed deprivation. Among households with similar income, the positional gap varies widely. This suggests that income-based identification, on its own, cannot determine which households are most deeply deprived, and that a depth-sensitive measure could usefully complement it.

\begin{figure}[H]
    \centering
    \begin{subfigure}[t]{0.48\linewidth}
      \centering
      \includegraphics[width=\linewidth]{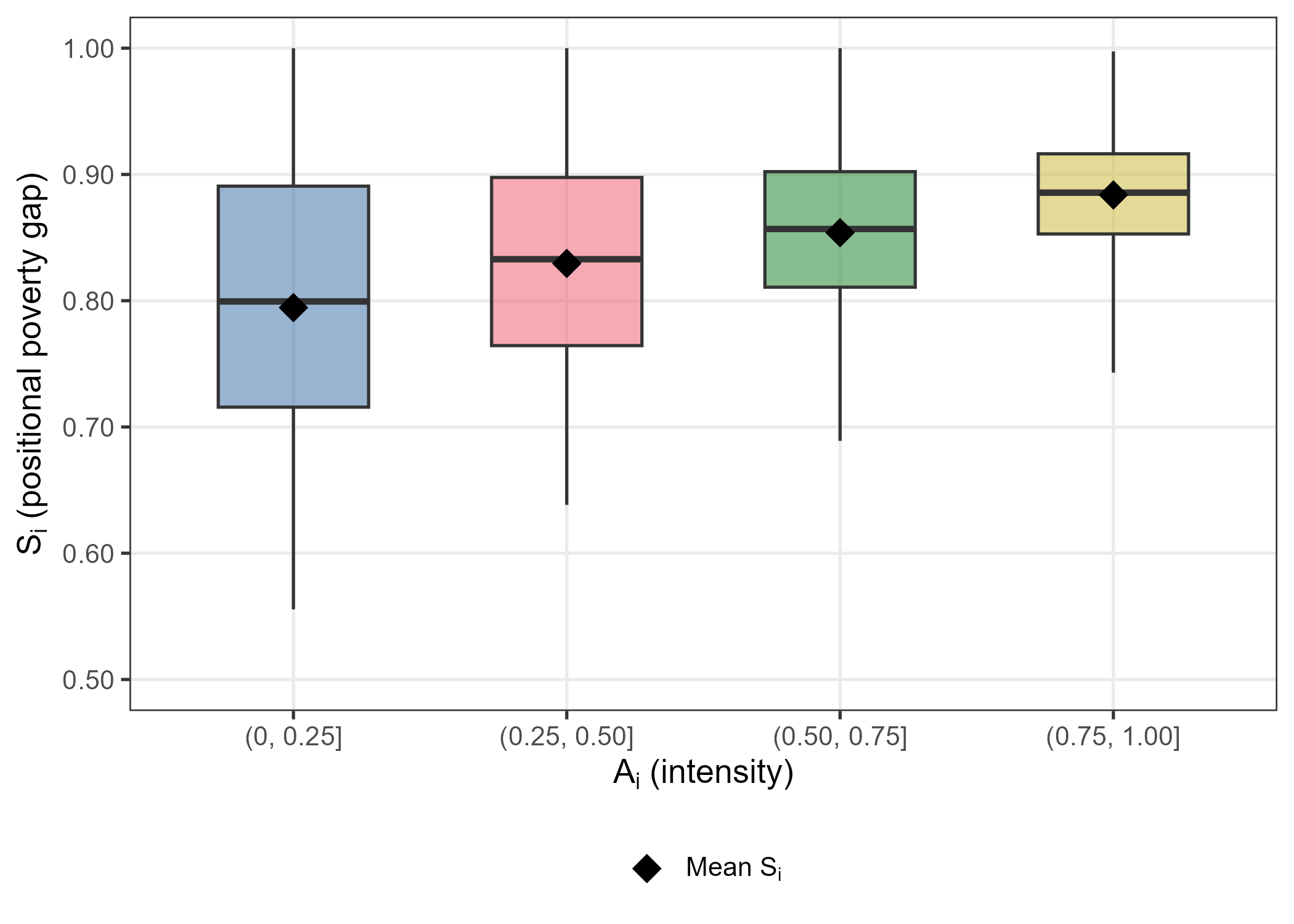}
      \caption{By deprivation intensity $A_i$}
      \label{fig:cad_box_A}
    \end{subfigure}%
    \hspace{0.04\linewidth}%
    \begin{subfigure}[t]{0.48\linewidth}
      \centering
      \includegraphics[width=\linewidth]{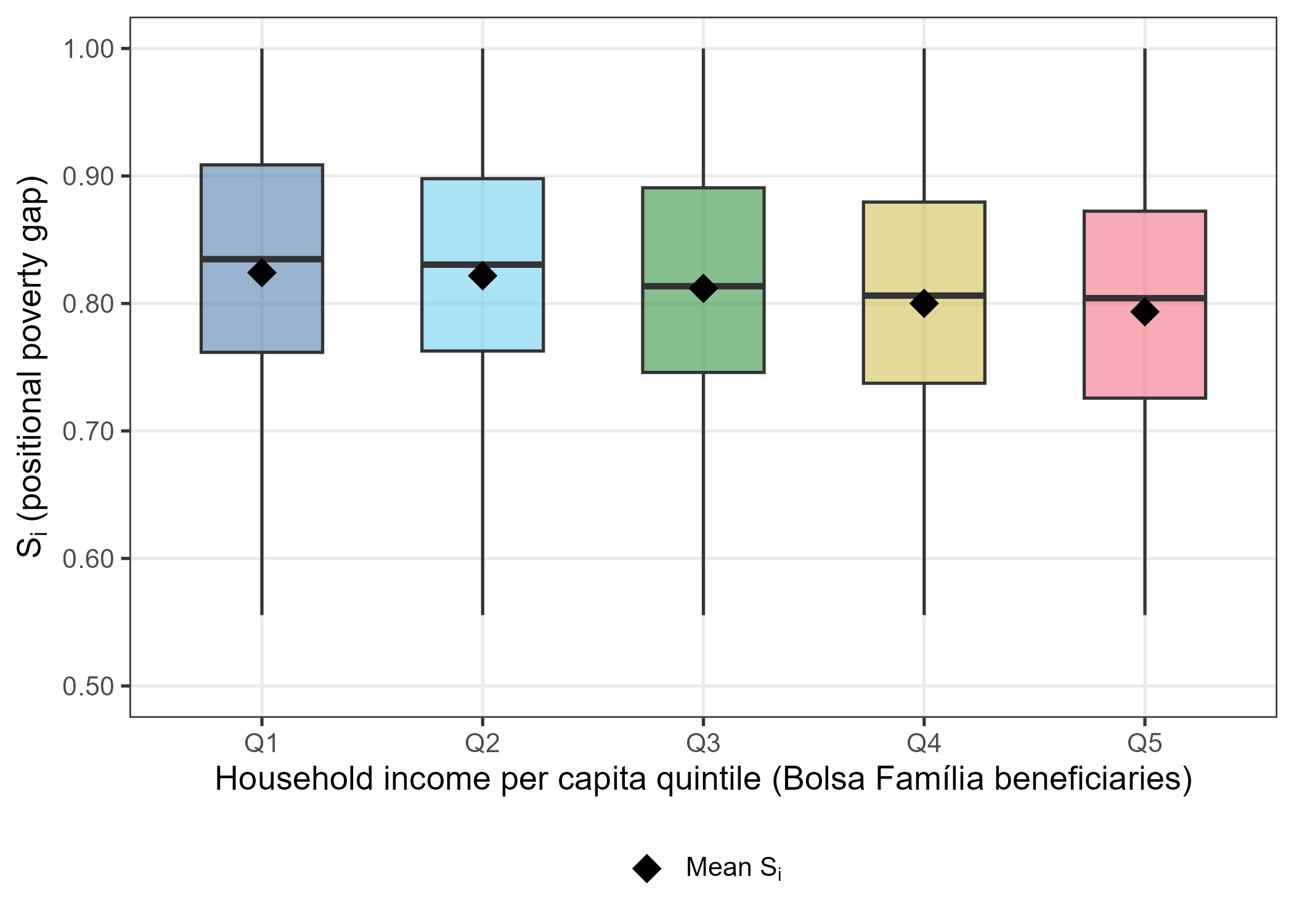}
      \caption{By income quintile, Bolsa Família beneficiaries}
      \label{fig:cad_box_incomeBF}
    \end{subfigure}
    \caption{Individual positional poverty gap $S_i$: distribution by deprivation intensity and by income, Brazil, Cadastro Único 2018.}
    \label{fig:cad_scatter_Si_Ai}
    \vspace{4pt}
    \begin{minipage}{\textwidth}
    \noindent {\notesize \textit{Notes:} Survey-weighted. Panel (a): households poor under the union cutoff, grouped by deprivation intensity $A_i$. Panel (b): households enrolled in Bolsa Família with at least one deprivation and non-missing income, grouped by income per capita quintile. Boxes show the weighted interquartile range, whiskers the weighted minimum and maximum, and diamonds the weighted mean $S_i$ within each bin.\par}
    \end{minipage}
\end{figure}

This illustration shows that the positional poverty gap displays robust dominance across subgroups and reveals depth that standard approaches miss. Breadth-based measures assign the same value to households with the same weighted deprivation score regardless of their depth. Income is only weakly related to depth, so households with the same income can differ sharply in deprivation depth.

\subsection{Cross-country MICS data}
 \label{section: illustration_mics}
 
As mentioned in Subsection \ref{section: correspondence_AFPG_PPG}, the AF poverty gap and the positional poverty gap measure different concepts, so their magnitudes are not directly comparable. For this reason, in this subsection I compare only their rankings. Within each dimension, the two measures rank individuals in the same way. The question is whether this agreement remains after several dimensions are combined into one score. A close association among the multiply-deprived would indicate that, in a given setting, the positional poverty gap produces individual orderings similar to those of the AF poverty gap, without implying equivalence between the two measures.

This section analyzes 23 of the 33 UNICEF Multiple Indicator Cluster Surveys (MICS) conducted between 2017 and 2023 that contain cardinal measures of three indicators commonly used in multidimensional poverty analysis \citep{unicef_mics}. The 23 retained countries have at least 200 individuals deprived in two or more indicators\footnote{The ten excluded countries have fewer than 200 multiply-deprived individuals; one has none, so its correlation is undefined, and the other nine range from 0.55 to 1.00.}. The indicators are education, measured in years of schooling; nutrition, measured by child anthropometric z-scores; and assets, measured as a count of household assets. For each indicator, the positional scores use CDFs computed separately within each country, over the individuals with a valid observation of that indicator. In Appendix~\ref{appendix: C}, Table \ref{table:C_defs} details the indicator definitions, and Table \ref{table:C_deprivrates} reports the survey-weighted deprivation rate of each indicator for all 33 countries, including the ten excluded from the analysis.

For each equally weighted indicator, I compute the normalized AF poverty gap scores $g^{AF}_{ij}$ and the positional depth scores $s_{ij}$. The AF poverty gap requires ratio-scale measurement, so I assume each cardinal comparator to be ratio-scale, directly for education and after the scale adjustments described in the notes to Table~\ref{table:C_defs} for assets and nutrition.
Because the two measures produce the same ranking within each dimension, I focus on individuals deprived in at least two of the three indicators ($k=2/3$), for whom the two measures aggregate deprivations of different types within the same profile, and compute $G_i$ and $S_i$ as the corresponding censored weighted averages over deprived indicators.

Figure~\ref{fig:forest_correspondence} reports the country-level weighted Spearman correlation between the two measures across the 23 countries. The cross-country median is approximately $0.83$, and in most countries the two measures agree closely on who is most and least deprived.

\begin{figure}[H]
	\centering
	\includegraphics[width=0.75\textwidth]{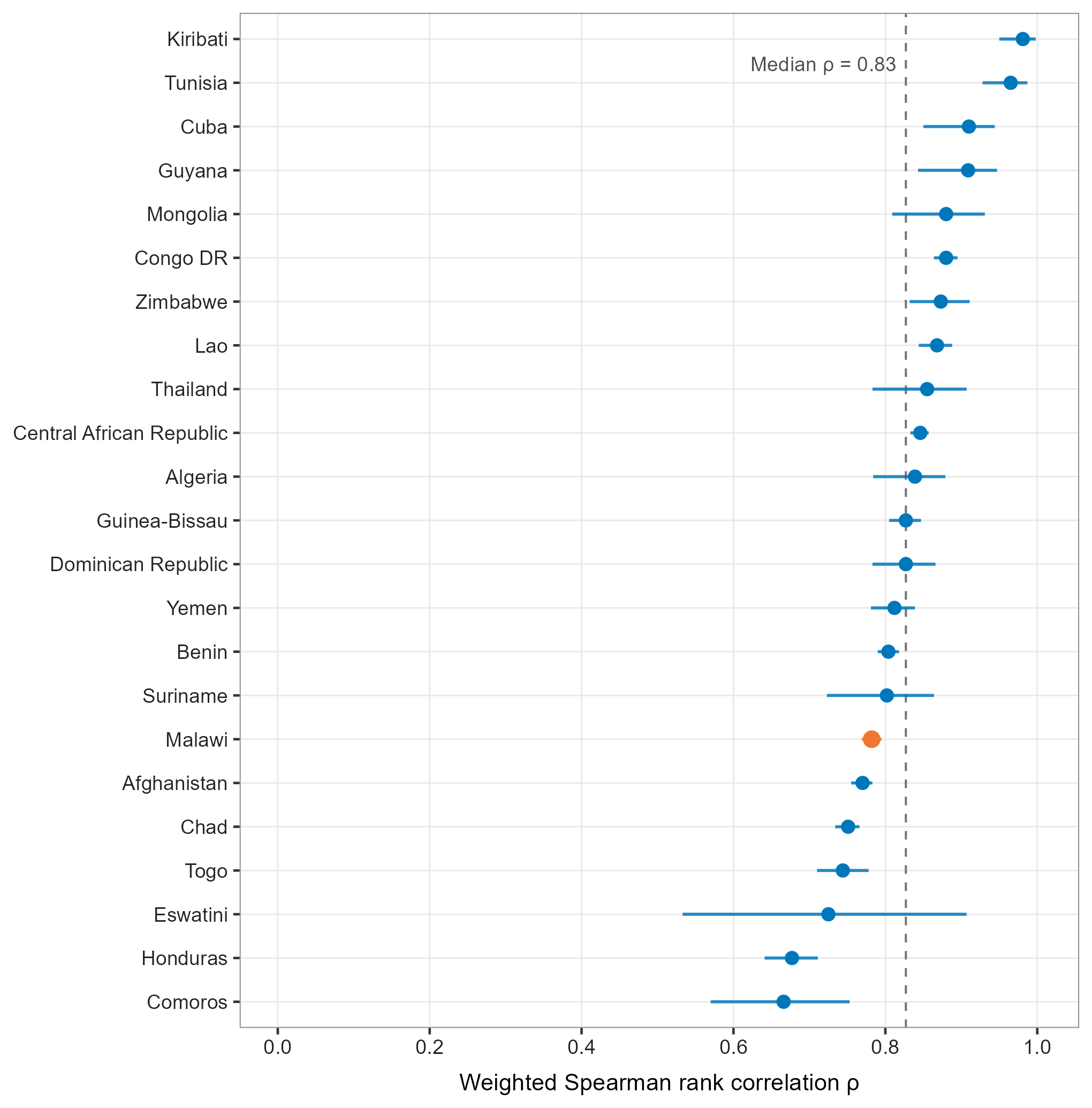}
	\caption{Country-level rank correspondence between the positional poverty gap and the Alkire--Foster poverty gap, MICS countries.}
	\label{fig:forest_correspondence}
	\vspace{4pt}
	\begin{minipage}{\textwidth}
		\noindent {\notesize \textit{Notes:} Each point is a country's weighted Spearman correlation between the individual positional and Alkire--Foster poverty gaps, among its multiply-deprived poor. Horizontal lines are 95\% confidence intervals from a stratified cluster bootstrap; the dashed line marks the cross-country median. Malawi (Figure~\ref{fig:malawi_rankdiff}) is shown in orange.\par}
	\end{minipage}
\end{figure}

This close agreement does not mean the two measures are identical. To illustrate how they can differ, Figure~\ref{fig:malawi_rankdiff} looks at one country, Malawi. For each person among the multiply-deprived poor, the figure compares the rank given by each measure. The horizontal axis is the difference in percentile rank, $r(S_i) - r(G_i)$, where $r$ is the weighted percentile rank within the group. Most people sit near zero, as they receive almost the same rank under both measures, which fits the high correlation. The rest divide into two tails, defined as percentile-rank differences below $-0.10$ and above $+0.10$. The concentration of disagreements in two tails on opposite sides of zero suggests that they are structured rather than random.

\begin{figure}[H]
	\centering
	\includegraphics[width=0.7\textwidth]{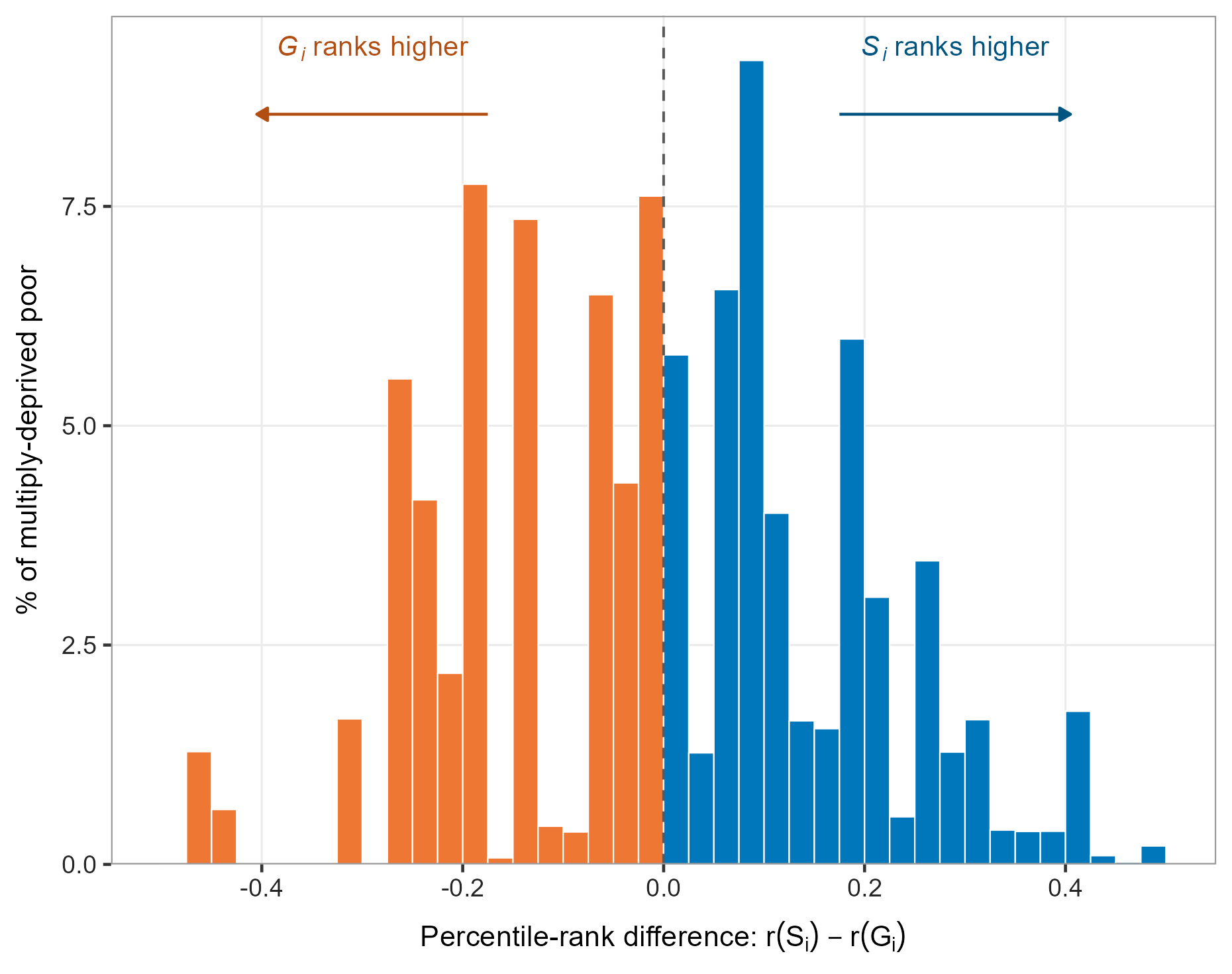}
	\caption{Distribution of individual rank differences between the positional poverty gap and the Alkire--Foster poverty gap, multiply-deprived poor, Malawi (MICS 2019--20).}
	\label{fig:malawi_rankdiff}
	\vspace{4pt}
	\begin{minipage}{\textwidth}
		\noindent {\notesize \textit{Notes:} Survey-weighted. Bars give the share of Malawi's multiply-deprived poor by percentile-rank difference $r(S_i) - r(G_i)$. Orange bars (left) are individuals ranked higher by the AF poverty gap; blue bars (right), higher by the positional poverty gap.\par}
	\end{minipage}
\end{figure}

These two tails correspond to different deprivation profiles. Table~\ref{tab:malawi_tails} reports, for each indicator, the deprivation rate in the full population, the share of individuals deprived in it among the poor and within each tail, and the mean and observed range of the AF gap and of the positional depth among the deprived. In the positive tail, where the positional poverty gap assigns the higher rank, every individual is deprived in nutrition ($100\%$, against $53\%$ of the poor overall), while in the negative tail only $2\%$ are. In the negative tail, where the AF poverty gap assigns the higher rank, $100\%$ of individuals are deprived in education, against $71\%$ among the poor overall. Asset deprivation, nearly universal in this sample ($94\%$), appears at high rates in both tails and so does not drive the divergence.

\begin{table}[H]
      \centering
      \caption{Composition of the rank-difference tails, Malawi (MICS 2019--20)}
      \footnotesize
      \setlength{\tabcolsep}{3pt}
      \begin{tabular}{@{}lcccccccc@{}}
        \toprule
                  & \multicolumn{4}{c}{\% deprived among:} & \multicolumn{2}{c}{AF gap} & \multicolumn{2}{c}{Positional depth} \\
        \cmidrule(lr){2-5} \cmidrule(lr){6-7} \cmidrule(lr){8-9}
        Indicator & Population & Poor & Neg.\ tail & Pos.\ tail & Mean & Min--max & Mean & Min--max \\
        \midrule
        Education & 28 & 71 & 100 & 58  & 0.46 & 0.17--1.00 & 0.86 & 0.74--1.00 \\
        Nutrition & 24 & 53 & 2   & 100 & 0.06 & 0.00--0.57 & 0.88 & 0.76--1.00 \\
        Assets    & 56 & 94 & 100 & 80  & 0.80 & 0.50--1.00 & 0.85 & 0.62--1.00 \\
        \bottomrule
      \end{tabular}
      \begin{tablenotes}
        \notesize
        \item \textit{Notes:} Survey-weighted. The population column reports each indicator's uncensored headcount ratio $h_j$ in equation~\eqref{eq:floor}; the next three columns report the share (percent) of multiply-deprived poor individuals deprived in each indicator, overall and within the negative and positive rank-difference tails of Figure~\ref{fig:malawi_rankdiff}. The last four columns report the mean and the observed range of the AF gap and of the positional depth among the deprived.
      \end{tablenotes}
      \label{tab:malawi_tails}
    \end{table}

The divergence arises from how the two measures evaluate nutrition relative to education once deprivations are combined into an individual-level score. On the positional scale, the two deprivations have similar depths (mean values of $0.86$ and $0.88$ among the deprived). Because about three quarters of Malawi's population lies above either cutoff, the floor in equation~\eqref{eq:floor} assigns every deprived individual a relatively high positional score. On the AF scale, by contrast, education is substantially deeper than nutrition. Nutrition-deprived individuals tend to lie only slightly below the cutoff, producing small normalized gaps (mean $0.06$), whereas education-deprived individuals are distributed across the full range below the cutoff and therefore have much larger gaps (mean $0.46$). Once these depths are averaged across an individual's deprived dimensions, the difference can reverse the ranking of deprivation profiles. A profile concentrated in nutrition receives a high positional score but a small AF gap and therefore tends to appear in the positive tail; a profile concentrated in education tends instead to appear in the negative tail.

The high positional depth for nutrition does not imply that a moderate undernutrition value is intrinsically or clinically severe. A z-score below $-2$ marks the threshold at which nutritional status is classified as undernutrition under the WHO growth standards \citep{who2006growth}. The cutoff therefore identifies a deprivation state, while the positional score records how exceptional that state is relative to others. In a population where three quarters are above the nutrition cutoff, being below it places an individual in a small minority that is disadvantaged relative to the rest, which is what a relative reading of depth registers. Nutrition-deprived individuals are more common in the positive tail in all 23
countries, while the individuals that fill the negative tail depend more on
each country's deprivation pattern.

These results suggest that the positional poverty gap tends to produce rankings closely associated with those of the AF poverty gap and may therefore be informative in settings where cardinal indicators are unavailable. At the same time, the measures capture different notions of depth: absolute shortfall from a normative standard versus relative disadvantage within the distribution of achievements. This mirrors a broader distinction in the poverty-measurement literature between absolute and relative conceptions, which can yield genuinely different assessments rather than more or less precise estimates of the same quantity \citep{Decerf2017, DecerfEtAl2025}. The two depth measures are therefore best read as complementary, not as interchangeable estimators of a single underlying quantity.

\section{Conclusion}
 \label{section: conclusion}

This paper proposes a positional poverty gap measure for capturing depth in ordinal settings, an area that has received relatively limited attention in multidimensional poverty measurement. Drawing on the distributional logic of the TFR approach within the AF counting framework, the measure interprets depth as the normalized distance of an individual’s relative position from the top of the distribution.

The measure can help distinguish individuals with relatively few but deeper deprivations and highlight subgroups with low incidence but high positional poverty gaps. Such patterns may not be fully captured by incidence and intensity alone, so positional depth can provide useful complementary information for identifying vulnerable groups and informing policy targeting.

The empirical applications illustrate these additional insights. In Brazil, the positional poverty gap reveals substantial heterogeneity among poor households that intensity alone does not capture. This variation persists among households with similar incomes, showing that income alone cannot distinguish which households experience deeper deprivation. The cross-country MICS comparison shows a high rank association between the positional poverty gap and the AF poverty gap, while showing that the two capture different aspects of depth. The AF poverty gap reflects absolute distance from the deprivation threshold, while the positional gap reflects relative position within the distribution, thereby capturing the context of deprivation.

The paper characterizes the proposed depth score and shows that it satisfies desirable axiomatic properties when the reference distribution is fixed over time. Comparisons across time or countries, however, require an explicit normative choice of reference distribution, and the paper discusses the tradeoffs associated with alternative benchmarks.

The framework fits naturally within existing AF applications, which may facilitate its use in empirical and policy analysis. In this sense, the positional poverty gap can complement existing measures by adding information on the context of disadvantages.

\newpage

\bibliographystyle{apalike} 
\bibliography{mybiblio2024_v4}
\addcontentsline{toc}{section}{\refname}





\newpage

\appendix

\section{Admissible transformations}
\label{appendix: A}
\setcounter{table}{0}
\renewcommand{\thetable}{A\arabic{table}}

\begin{table}[ht]
	\centering
	\caption{Admissible transformations for variables that are relevant for poverty measurement}
	\footnotesize
	\setlength{\tabcolsep}{2pt}
	\begin{tabular}{p{5.5em}p{6.285em}p{6.215em}p{9.6em}p{8.645em}}
		\toprule
		\textbf{Type of variable} & \textbf{Admissible transformations} & \textbf{Admissible mathematical operations} & \textbf{Permissible statistics} & \textbf{Examples relevant for poverty measurement} \\
		\midrule
		Qualitative: Ordinal & Order-preserving transformations (monotonic increasing functions). & \multicolumn{1}{l}{None} & Frequency distribution, mode, contingency, correlation, median percentiles & Type of sanitation facility, source of drinking water, cooking fuel, floor, walls, and roof; levels of schooling. \\
		Cardinal: Interval scale & linear transformations & \multicolumn{1}{l}{Add, subtract} & Frequency distribution, mode, contingency, correlation, median percentiles, mean, standard deviation, rank-order, correlation, product-moment correlation. & z-scores of nutritional indicators (e.g., weight-for-age or body-mass-index-for-age). \\
		Cardinal: Ratio scale & linear and ratio transformations & Add, subtract, divide, multiply & Frequency distribution, mode, contingency, correlation, median percentiles, mean, standard deviation, rank-order, correlation, product-moment correlation, coefficient of\newline{}variation & Income, consumption expenditure, body mass index, number of deaths a mother experienced, years of schooling, numbers of bedrooms, number of assets (e.g., car, computer, fridge, and others). \\
		\hline \hline
	\multicolumn{5}{l}{Source: Adapted from \citet{Stevens1946} and \citet{AlkireEtAl2015}.}
	\end{tabular}%
	\label{table:A1}%
\end{table}%

\newpage

\section{Proofs}
\label{appendix: B}
\textbf{Proof of Theorem~\ref{thm:rd_foundation} (characterization of the positional depth score).}
Additive decomposition applied with an empty comparison group gives
$\delta_{ij}(X_j,\emptyset)=2\,\delta_{ij}(X_j,\emptyset)$, hence $\delta_{ij}(X_j,\emptyset)=0$.
By focus and repeated additive decomposition,
\[
\delta_{ij}(X_j,\mathcal{R})\;=\;\sum_{h\in \mathcal{R}:\;x_{hj}>x_{ij}} \delta_{ij}(X_j,\{h\}).
\]
By pairwise comparison, each term equals $f(x_{ij},x_{hj})$ for a single function $f$;
by focus, only its values on ordered pairs $\{(u,v):u<v\}$ are relevant. By ordinal invariance,
$f(\varphi(u),\varphi(v))=f(u,v)$ for every strictly increasing $\varphi$ of the
real line. For any two ordered pairs $u<v$ and $u'<v'$ there exists a strictly
increasing bijection of the real line, piecewise linear for instance, mapping $u$ to $u'$ and $v$ to $v'$;
hence $f$ is constant on this domain, $f\equiv \lambda$; non-triviality and nonnegativity give $\lambda> 0$. Therefore
\[
\delta_{ij}(X_j,\mathcal{R})\;=\;\lambda\cdot\big|\{h\in \mathcal{R}:\;x_{hj}>x_{ij}\}\big|,
\]
proportional to the population share $1-F_j(x_{ij})$\footnote{The conditions and the proof are stated on the reference population itself, so that
$F_j$ is the population distribution and the count is over population members. With survey data,
the weighted sample shares estimate the population shares $1-F_j(x_{ij})$; weights are an
estimation device and play no role in the characterization.}. Since
$F_j(\min_i x_{ij})<1$, the most deprived member of $\mathcal{R}$ has a
nonzero deprivation value, and the normalized score is
\[
s_{ij}\;=\;\frac{\lambda\,\big|\{h\in \mathcal{R}:\;x_{hj}>x_{ij}\}\big|}{\lambda\,\big|\{h\in \mathcal{R}:\;x_{hj}>\min_i x_{ij}\}\big|}
\;=\;\frac{1-F_j(x_{ij})}{1-F_j(\min_i x_{ij})},
\]
in which the constant and the population size cancel. Conversely, for any $\lambda> 0$ the deprivation function
$\delta_{ij}(X_j,B)=\lambda\,\big|\{h\in B:\;x_{hj}>x_{ij}\}\big|$ satisfies pairwise comparison, focus,
additive decomposition, and ordinal invariance, and induces the score above when evaluated at
$B=\mathcal{R}$ and normalized. \qed

\medskip

\emph{Independence of the conditions.} Each condition excludes a distinct family of
alternatives, and no condition is implied by the others. Without pairwise comparison,
distribution-dependent bilateral terms survive, such as
$\delta_{ij}(X_j,\{h\})=\mathbf{1}\{x_{hj}>x_{ij}\}\,\big(1-F_j(x_{hj})\big)$, which
weights each better-off comparator by that person's own position while satisfying
focus, additive decomposition, and ordinal invariance. Without additive
decomposition, any increasing transformation of the count survives: the deprivation function
$\delta_{ij}(X_j,B)=\big|\{h\in B:\;x_{hj}>x_{ij}\}\big|^{\gamma}$ with $\gamma>0$ and $\gamma\neq 1$
satisfies the remaining conditions and induces the normalized score
$\big[(1-F_j)/(1-F_j(\min_i x_{ij}))\big]^{\gamma}$; additivity
fails because $(a+b)^{\gamma}\neq a^{\gamma}+b^{\gamma}$ for comparison groups of
sizes $a,b\ge 1$. Without ordinal invariance, the elementary comparison may depend on
the categories involved, admitting category-weighting schemes in the spirit of
\citet{SethYalonetzky2021}; replacing ordinal invariance with translation invariance
and linear homogeneity yields Yitzhaki's cardinal index up to normalization
\citep{EbertMoyes2000}. Without the equality part of focus, the score that also counts
individuals at the same position, proportional to the share with $x_{hj}\ge x_{ij}$,
satisfies the remaining conditions. Without the
normalization of Section~\ref{section: distributional_poverty_gap}, any positive
multiple of the share survives.

\bigskip

\textbf{Proofs of Theorem~\ref{thm:theorem} (Axiomatic properties).}

\textbf{Deprivation Focus.}
At $(i,\theta)$, $g^0_{i\theta}=0$ both before and after, hence $g^1_{i\theta}(k)=\rho_i(k)\,s_{i\theta}\,g^0_{i\theta}=0$ in both matrices.
For any deprived peer $(p,\theta)$ with $x_{p\theta}<z_\theta$, the CDF term $F_\theta(x_{p\theta})$ is the (weighted) \emph{share} of the
reference distribution, in-sample or anchored, at or below $x_{p\theta}$. Since $x_{i\theta},x'_{i\theta}\ge z_\theta>x_{p\theta}$, this share is
unchanged, so $s_{p\theta}$ (and thus $g^1_{p\theta}(k)$) is unchanged. The normalizer $1-F_\theta(\min_i x_{i\theta})$ is also unchanged:
with in-sample CDFs a change above $z_\theta$ cannot alter the minimum among deprived values (and if no one is deprived,
the column contributes zero); with anchored CDFs the reference minimum is fixed. Therefore $P(X')=P(X)$.

\bigskip

\textbf{Poverty Focus.}
\emph{Anchored CDFs.} For the non-poor person $i^\star$, $\rho_{i^\star}(k)=\rho'_{i^\star}(k)=0$, so
$g^1_{i^\star j}(k)=\rho_{i^\star}(k)\,s_{i^\star j}\,g^0_{i^\star j}=0$ and
$g^{1\,\prime}_{i^\star j}(k)=\rho'_{i^\star}(k)\,s'_{i^\star j}\,g^{0\,\prime}_{i^\star j}=0$
for all $j$, regardless of $x'_{i^\star j}$. For every other $(p,j)$ with $p\neq i^\star$,
anchoring keeps $s_{pj}$ fixed; and by hypothesis $x'_{pj}=x_{pj}$, hence
$g^{1\,\prime}_{pj}(k)=g^1_{pj}(k)$. Therefore every summand in
$P(X;z,k)=(1/n)\sum_{i,j} w_j g^1_{ij}(k)$ is unchanged and $P(X')=P(X)$.

\emph{In-sample CDFs under union identification.}
Here $\rho_{i^\star}(k)=0$ implies $g^0_{i^\star j}=0$ for all $j$, i.e.\ $x_{i^\star j}\ge z_j$ in every
indicator, and the same holds after the change. For any indicator $j$ and any deprived peer $p$ with
$x_{pj}<z_j$, the empirical CDF $F_j^{\mathrm{IS}}(x_{pj})$ equals the (weighted) share of observations
$\le x_{pj}$. Because $x_{i^\star j},x'_{i^\star j}\ge z_j > x_{pj}$, that share at $x_{pj}$ is unchanged,
so the numerators $1-F_j^{\mathrm{IS}}(x_{pj})$ are unchanged for all deprived peers. The denominator
$1-F_j^{\mathrm{IS}}(m_j)$, with $m_j=\min_i x_{ij}$, is also unchanged: if there exists at least one
deprived observation in $j$, then $m_j<z_j$ and altering a non-deprived value cannot change the minimum;
if there are no deprived observations, then $g^0_{pj}(k)\equiv 0$ and the column contributes zero both
before and after. Hence all $g^1_{pj}(k)$ are unchanged for every $(p,j)$, and the terms for $i^\star$
are zero both before and after. Consequently, $P(X')=P(X)$.

\bigskip

\textbf{Monotonicity.} 
\textit{Own-person.} Consider a worsening for person $i$ in indicator $\theta$.

\emph{Case (i):} $x'_{i\theta}<x_{i\theta}<z_\theta$ (still deprived). Then $g^0_{i\theta}$ and $\rho_i(k)$
are unchanged, and only $s_{i\theta}$ moves. Under \emph{anchored} CDFs, the reference $F_\theta$ and its minimum
are fixed and unaffected by the change. Since $F_\theta$ is nondecreasing,
$x'_{i\theta}<x_{i\theta}\Rightarrow F_\theta(x'_{i\theta})\le F_\theta(x_{i\theta})$, hence $s'_{i\theta}\ge s_{i\theta}$;
if $x'_{i\theta}$ falls below the reference support, the convention scores it at one, the maximum value.
Under \emph{in-sample} CDFs, both $F_\theta$ and, possibly, the column minimum $m_\theta$ are recomputed after
the change, so the comparison proceeds by counting individuals. The new numerator $1-F'_\theta(x'_{i\theta})$ is
the share of individuals with achievements strictly above $x'_{i\theta}$ in the new sample; because only person
$i$ has moved, and moved downward, this set contains every individual counted in the old numerator
$1-F_\theta(x_{i\theta})$, so the numerator cannot fall. For the normalizer, if $x'_{i\theta}>m_\theta$ the mass
at or below $m_\theta$ is unchanged, so the normalizer is unchanged and $s'_{i\theta}\ge s_{i\theta}$; if
$x'_{i\theta}\le m_\theta$, person $i$ becomes (one of) the most deprived in the column and
$s'_{i\theta}=1\ge s_{i\theta}$. In both implementations the inequality is strict under the positive-mass
condition of the statement: if $s_{i\theta}<1$ and the reference distribution places positive mass in
$(x'_{i\theta},x_{i\theta}]$, contributed by individuals other than $i$ in the in-sample case, since the own
mass moves with the change, then $s'_{i\theta}>s_{i\theta}$. Therefore
\[
\Delta P_i
= \rho_i(k)\,w_\theta\big(s'_{i\theta}-s_{i\theta}\big)\;\ge\;0,
\]
strictly positive when $\rho_i(k)=1$ under the positive-mass condition, while for a non-poor person all cells are censored and $\Delta P_i=0$.

\emph{Case (ii):} $x'_{i\theta}<z_\theta\le x_{i\theta}$ (crosses into deprivation). Here $g^0_{i\theta}$
switches from $0$ to $1$. If $i$ was already poor, $\rho_i(k)=\rho'_i(k)=1$ and a single new non-negative
term appears:
\[
\Delta P_i
= w_\theta\,s'_{i\theta}\;\ge\;0.
\]
If the change makes $i$ newly poor ($\rho_i(k)=0$, $\rho'_i(k)=1$), censoring is lifted from all of $i$'s
deprived indicators, so
\[
\Delta P_i
= w_\theta\,s'_{i\theta} \;+\; \sum_{j\neq\theta} w_j\,s_{ij}\,g^0_{ij}\;\ge\;0,
\]
where every added term is non-negative; if $i$ remains non-poor ($\rho'_i(k)=0$), $\Delta P_i=0$. The change
is strictly positive whenever $\rho'_i(k)=1$ and $F_\theta(x'_{i\theta})<1$. Hence dimensional monotonicity
also holds.

In all cases, a worsening cannot reduce $P_i$.

\medskip

\textit{Aggregate (anchored CDFs).} Let $X'$ be obtained by worsening cell $(i,\theta)$.
With anchored $\{F_j\}$, the positional scores $s_{pj}$ are unchanged for all $(p,j)\neq(i,\theta)$.
The censored contribution at $(i,\theta)$ can therefore change directly through
$g^1_{i\theta}(k)=\rho_i(k)\,g^0_{i\theta}\,s_{i\theta}$, and, if the worsening makes $i$ newly poor,
so can $i$'s other deprived cells $g^1_{ij}(k)$, $j\neq\theta$, through the lifting of censoring.
Consider two parallel cases.

\emph{Case (i):} $x'_{i\theta}<x_{i\theta}<z_\theta$ (still deprived).  
If $i$ is poor both before and after, $\rho_i(k)=1$ and
\[
\Delta P
= P(X';z,k)-P(X;z,k)
= \frac{1}{n}\,w_\theta\big(s'_{i\theta}-s_{i\theta}\big) \;\ge\; 0,
\]
with strict inequality under the positive-mass condition of the own-person case.
If $i$ is non-poor both before and after, then $\rho_i(k)=\rho'_i(k)=0$ and $\Delta P=0$
(by poverty focus).

\emph{Case (ii):} $x'_{i\theta}<z_\theta\le x_{i\theta}$ (crosses into deprivation).
If $i$ was already poor, a single new term enters and $\Delta P=\frac{1}{n}\,w_\theta\,s'_{i\theta}\ge 0$.
If the change makes $i$ newly poor, censoring is lifted from all of $i$'s deprived indicators, so
\[
\Delta P
= \frac{1}{n}\Big(w_\theta\,s'_{i\theta} \;+\; \sum_{j\neq\theta} w_j\,s_{ij}\,g^0_{ij}\Big)\;\ge\;0,
\]
with every added term non-negative; if $i$ remains non-poor, $\Delta P=0$. The change is strictly positive
whenever $\rho'_i(k)=1$ and $F_\theta(x'_{i\theta})<1$. Therefore aggregate dimensional
monotonicity holds under anchored CDFs.

In all cases, a worsening cannot reduce $P_i$; and under anchored CDFs the aggregate
index $P(X;z,k)$ weakly increases (strictly when a poor person worsens in a deprived
indicator or becomes newly deprived and poor, under the positive-mass condition above).
For $P_\alpha$ with $\alpha>1$, the same arguments apply with $s_{i\theta}$ replaced by
$s_{i\theta}^{\alpha}$: since $t\mapsto t^{\alpha}$ is strictly increasing on $[0,1]$,
every inequality above, and its strictness, is preserved.

\bigskip
\textbf{Subgroup Consistency and Decomposability.} Let the population be partitioned into
disjoint subgroups $\ell_1,\ldots,\ell_L$ with sizes $n^\ell$ (so $\sum_\ell n^\ell=n$), and let
$X^\ell$ denote the restriction of $X$ to subgroup $\ell$. Because the CDFs are anchored and
$(z,k,w)$ are common, each cell term $g^{\alpha}_{ij}(k)$ takes the \emph{same} value whether
computed in the full sample or within any subgroup.

\emph{(i) Subgroup consistency.}
Let $X'$ differ from $X$ only within subgroup $\ell'$, and assume subgroup sizes are unchanged.
Then, grouping the average by subgroups,
\[
P(X';z,k)-P(X;z,k)
=\frac{n^{\ell'}}{n}\Big(P(X'^{\,\ell'};z,k)-P(X^{\ell'};z,k)\Big).
\]
Hence the overall index moves in the \emph{same direction} as the index of the changed subgroup:
if $P(X'^{\,\ell'};z,k)<P(X^{\ell'};z,k)$, then $P(X';z,k)<P(X;z,k)$ (and conversely if it worsens).

\emph{(ii) Population subgroup decomposability.} Grouping the population sum by subgroups,
\[
\begin{aligned}
P(X;z,k)
&=\frac{1}{n}\sum_{\ell=1}^{L}\sum_{i\in \ell}\sum_{j=1}^{d} w_j\,g^{\alpha}_{ij}(k)
 =\sum_{\ell=1}^{L}\frac{n^\ell}{n}\Bigg(\frac{1}{n^\ell}\sum_{i\in \ell}\sum_{j=1}^{d} w_j\,g^{\alpha}_{ij}(k)\Bigg)\\
&=\sum_{\ell=1}^{L}\frac{n^\ell}{n}\;P(X^\ell;z,k),
\end{aligned}
\]
where the last equality uses that, under the common (anchored) CDFs, the parenthesised term equals
the subgroup index $P(X^\ell;z,k)$.

\emph{Weighted analogue.} With sampling weights, replace $n^\ell/n$ by $W^\ell/W$, where
$W^\ell=\sum_{i\in \ell}\omega_i$ and $W=\sum_{i=1}^n\omega_i$.

\bigskip
\textbf{Remaining properties.} Symmetry, replication invariance, bounds and normalization, ordinal invariance, and weak rearrangement follow directly from the construction of $s_{ij}$ and the structure of $P$, as established in Section~\ref{section: properties}. \qed

\bigskip
\textbf{Proof of Proposition~\ref{prop:within-dim}.}
On the deprived range $x_{ij}<z_j$, the condition $z_j>0$ implies that the
Alkire--Foster gap $g^{AF}_{ij}$ is strictly decreasing in $x_{ij}$.
Let $a$ and $b$ be two individuals deprived in dimension $j$. If $x_{aj}<x_{bj}$, then
$g^{AF}_{aj}>g^{AF}_{bj}$. Moreover, by the assumption on $F_j$,
$F_j(x_{aj})<F_j(x_{bj})$; since $1-F_j(\min_i x_{ij})>0$, it follows that
$s_{aj}>s_{bj}$. If $x_{aj}=x_{bj}$, then both scores are equal for the two individuals.
Thus $g^{AF}_{ij}$ and $s_{ij}$ generate the same within-dimension ranking of deprived
observations. The argument uses only the strict monotonicity of the cardinal comparator
on the deprived range, so it applies unchanged to any cardinal deprivation score that is
strictly decreasing in achievement there. \qed

\newpage
\section{Indicator definitions, thresholds, and deprivation rates}
\label{appendix: C}
\setcounter{table}{0}
\renewcommand{\thetable}{C\arabic{table}}

Table~\ref{table:C_defs} lists the indicators, deprivation thresholds, and weights of the two applications in Section~\ref{section: illustrations}: the eight indicators of the Brazil Cadastro Único index and the three cardinal indicators of the cross-country MICS comparison.

\begin{table}[H]
  \centering
  \caption{Indicators, deprivation thresholds, and weights, by application}
  \footnotesize
  \setlength{\tabcolsep}{3pt}
  \begin{tabular}{@{}l p{2.1cm} p{2.9cm} l p{3.3cm} c@{}}
    \toprule
    Dimension & Indicator & Measure & Scale & Deprivation threshold & Weight \\
    \midrule
    \multicolumn{6}{@{}l}{\textit{Panel A. Brazil, Cadastro Único 2018}} \\
    \addlinespace[2pt]
    Education & Adult attainment & Highest years of schooling among adults 18+ & Cardinal & Below 8 completed years & $\tfrac{1}{6}$ \\
    \addlinespace[1pt]
    Education & Child age-for-grade gap & Years behind expected grade, largest among children 6--17 & Cardinal & At least one year behind & $\tfrac{1}{6}$ \\
    \addlinespace[1pt]
    Housing & Housing materials & Inadequate floor or wall material & Ordinal & At least one inadequate component & $\tfrac{1}{6}$ \\
    \addlinespace[1pt]
    Housing & Crowding & Persons per bedroom & Cardinal & More than two persons per bedroom & $\tfrac{1}{6}$ \\
    \addlinespace[1pt]
    Services & Drinking water & Water source and piping ladder & Ordinal & Below piped supply from the general network & $\tfrac{1}{12}$ \\
    \addlinespace[1pt]
    Services & Sanitation & Toilet facility and disposal ladder & Ordinal & Below a septic-tank toilet & $\tfrac{1}{12}$ \\
    \addlinespace[1pt]
    Services & Garbage disposal & Refuse-disposal type & Ordinal & Not collected (burned, buried, or dumped) & $\tfrac{1}{12}$ \\
    \addlinespace[1pt]
    Services & Electricity & Lighting source & Ordinal & No electric lighting & $\tfrac{1}{12}$ \\
    \midrule
    \multicolumn{6}{@{}l}{\textit{Panel B. Cross-country MICS}} \\
    \addlinespace[2pt]
    Education & Years of schooling & Highest years of schooling in the household & Cardinal & No member has completed six years of schooling & $\tfrac{1}{3}$ \\
    \addlinespace[1pt]
    Nutrition & Nutritional status & Lowest child weight- or height-for-age $z$-score & Cardinal & Child $z$-score below $-2$ & $\tfrac{1}{3}$ \\
    \addlinespace[1pt]
    Assets & Asset ownership & Count of assets owned & Cardinal & Fewer than two assets & $\tfrac{1}{3}$ \\
    \bottomrule
  \end{tabular}
  \begin{tablenotes}
    \notesize
    \item \textit{Notes:} Weights are given in the last column and sum to one across dimensions. Education and nutrition thresholds follow the global MPI. Variables for which higher raw values indicate worse conditions (persons per bedroom, years behind expected grade) are sign-reversed, so that higher values always represent better achievements. In the MICS comparison (Section~\ref{section: illustration_mics}), the AF poverty gap is the standard normalized shortfall from the threshold, $(z_j-x_{ij})/z_j$, except where the indicator's scale requires adjustment: for assets, the gap is the normalized shortfall in the number of assets owned below the two-asset threshold; for nutrition, whose $z$-scores are interval-scale and have no natural zero, the shortfall is normalized by a fixed band of four $z$-score units, $(-2-x_{ij})/4$, capped at one below $-6$, the lower bound of biologically plausible weight- and height-for-age $z$-scores under the WHO growth standards \citep{who_unicef_biv2026}. The positional depth score requires no such floor.
  \end{tablenotes}
  \label{table:C_defs}
\end{table}

Table~\ref{table:C_deprivrates} reports the survey-weighted deprivation rate of each indicator in each of the 33 countries, together with the corresponding poverty measures.

\begin{table}[H]
  \centering
  \caption{Deprivation rates by indicator, MICS countries}
  \footnotesize
  \setlength{\tabcolsep}{3pt}
  \begin{tabular}{lcrccccc}
    \toprule
    Country & Survey & $n$ & Education & Nutrition & Assets & $H$ & $P$ \\
    \midrule
Afghanistan & 2022-23 & 197,062 & 41.5 & 50.2 & 36.7 & 0.412 & 0.266 \\
Algeria & 2019 & 145,890 & 5.1 & 5.8 & 0.6 & 0.007 & 0.005 \\
Argentina & 2019-20 &  48,285 & 2.1 & 4.7 & 0.7 & 0.002 & 0.001 \\
Azerbaijan & 2023 &  40,323 & 0.4 & 2.4 & 0.4 & 0.000 & 0.000 \\
Benin & 2021-22 &  82,534 & 41.5 & 30.7 & 26.6 & 0.283 & 0.182 \\
Central African Republic & 2019 &  44,182 & 46.3 & 44.8 & 70.8 & 0.571 & 0.365 \\
Chad & 2019 & 111,539 & 58.4 & 45.7 & 46.7 & 0.518 & 0.339 \\
Comoros & 2022 &  30,952 & 13.0 & 16.3 & 23.5 & 0.111 & 0.073 \\
Congo DR & 2017-18 & 102,737 & 16.6 & 42.4 & 54.3 & 0.349 & 0.221 \\
Cuba & 2019 &  38,356 & 2.2 & 2.0 & 6.5 & 0.011 & 0.007 \\
Dominican Republic & 2019 &  94,040 & 7.9 & 3.0 & 4.5 & 0.020 & 0.013 \\
Eswatini & 2021-22 &   9,628 & 4.8 & 15.8 & 13.1 & 0.039 & 0.024 \\
Georgia & 2018 &  34,988 & 0.3 & 1.9 & 1.0 & 0.001 & 0.000 \\
Guinea-Bissau & 2019 &  48,652 & 41.1 & 33.0 & 16.2 & 0.250 & 0.151 \\
Guyana & 2019-20 &  23,869 & 2.9 & 6.3 & 6.5 & 0.015 & 0.010 \\
Honduras & 2019 &  76,863 & 11.3 & 10.4 & 10.0 & 0.061 & 0.041 \\
Kiribati & 2019 &  17,346 & 0.4 & 14.3 & 24.8 & 0.045 & 0.028 \\
Kyrgyzstan & 2023 &  27,105 & 0.3 & 9.5 & 0.9 & 0.001 & 0.001 \\
Lao & 2023 &  90,576 & 20.7 & 19.7 & 6.9 & 0.090 & 0.060 \\
Malawi & 2019-20 & 107,695 & 28.1 & 24.3 & 55.7 & 0.322 & 0.202 \\
Mongolia & 2018 &  48,361 & 3.4 & 6.1 & 1.5 & 0.007 & 0.005 \\
Montenegro & 2018 &  10,454 & 2.2 & 2.4 & 0.1 & 0.001 & 0.001 \\
Samoa & 2019-20 &  19,709 & 0.2 & 10.4 & 7.3 & 0.012 & 0.008 \\
Serbia & 2019 &  16,997 & 2.0 & 0.9 & 0.1 & 0.000 & 0.000 \\
Suriname & 2018 &  28,918 & 7.1 & 5.0 & 4.2 & 0.022 & 0.015 \\
Thailand & 2021-22 &  95,423 & 10.5 & 2.8 & 1.1 & 0.006 & 0.004 \\
Togo & 2017 &  34,522 & 20.9 & 21.3 & 21.9 & 0.152 & 0.099 \\
Tonga & 2019 &  12,671 & 0.2 & 2.0 & 2.2 & 0.001 & 0.000 \\
Tunisia & 2023 &  31,088 & 10.9 & 3.6 & 1.0 & 0.008 & 0.005 \\
Turkmenistan & 2019 &  30,752 & 0.0 & 7.1 & 0.0 & 0.000 & 0.000 \\
Tuvalu & 2019-20 &   4,025 & 0.3 & 7.6 & 0.6 & 0.001 & 0.001 \\
Yemen & 2022-23 &  43,635 & 9.1 & 38.6 & 24.1 & 0.159 & 0.104 \\
Zimbabwe & 2019 &  42,885 & 3.9 & 17.4 & 29.1 & 0.082 & 0.050 \\
    \bottomrule
  \end{tabular}
  \begin{tablenotes}
    \notesize
    \item \textit{Notes:} Survey-weighted share of the population deprived in each indicator (uncensored headcount ratio, in percent), using the global MPI deprivation thresholds; $n$ is the sample size. $H$ is the population share poor at the $k=2/3$ cutoff of Section~\ref{section: illustration_mics} and $P = H \cdot A \cdot S$ the corresponding positional poverty gap index, both reported as proportions.
  \end{tablenotes}
  \label{table:C_deprivrates}
\end{table}

\end{document}